\documentclass[journal,10pt,final,letterpaper]{IEEEtran}
\usepackage[pdftex]{graphicx}
\usepackage[cmex10]{amsmath}
\usepackage[tight,footnotesize]{subfigure}
\usepackage{url}
\usepackage{amsfonts}
\usepackage{amssymb}
\usepackage{amsthm}
\usepackage{bm}
\usepackage{array}
\usepackage{multicol}
\usepackage[noadjust]{cite}
\usepackage{enumitem}
\usepackage{footmisc}
\usepackage{url}
\usepackage[colorlinks, bookmarks=false, citecolor=blue, linkcolor=red, urlcolor=blue]{hyperref}
\newtheorem{proposition}{Proposition}
\newtheorem{assumption}{Assumption}
\newtheorem{remark}{Remark}

\newcommand{\dd}{\mathop{}\!\mathrm{d}}

\let\originalleft\left
\let\originalright\right
\def\left#1{\mathopen{}\originalleft#1}
\def\right#1{\originalright#1\mathclose{}}
\newcommand\norm[1]{\left\lVert#1\right\rVert}

\begin{document}

\title{Performance Measures in Electric Power Networks under Line Contingencies}
\author{Tommaso~Coletta, and Philippe~Jacquod,~\IEEEmembership{Member,~IEEE}
\thanks{T. Coletta and Ph. Jacquod are with the School of Engineering of the University of Applied Sciences of Western Switzerland
CH-1951 Sion, Switzerland. Emails: (tommaso.coletta, philippe.jacquod)@hevs.ch.}
}

\markboth{THIS VERSION: April 17 2019}%
{Coletta \MakeLowercase{\textit{et al.}}: Performance Measures in Electric Power Networks under Line Contingencies}

\maketitle

\begin{abstract}
Classes of performance measures expressed in terms of ${\cal H}_2$-norms 
have been recently introduced to quantify the response of
coupled dynamical systems to external perturbations. So far,
investigations of these performance measures have been restricted to nodal perturbations.
Here, we go beyond these earlier works and consider the equally important, but so far neglected case of line perturbations.
We consider a network-reduced power system, where a Kron reduction has eliminated passive buses. Identifying
the effect that a line fault in the physical network has on the Kron-reduced network, we find that performance measures
depend on whether the faulted line connects two passive, two active buses or one active to one passive bus. In
all cases, performance measures depend quadratically on the original load on the faulted line
times a topology dependent factor. Our theoretical 
formalism being restricted to Dirac-$\delta$  perturbations, we investigate numerically the validity of our results for finite-time 
line faults.  For uniform damping over inertia ratios,
we find good agreement with theoretical predictions for longer fault durations in systems with more inertia, for which 
eigenmodes of the network are harder to excite.
\end{abstract}


\section{Introduction}

In normal operation, electric power grids are synchronized. Their operating state corresponds to equal frequencies and
voltage angle differences ensuring power conservation at all buses. Such synchronous states are not specific to power grids. 
They occur in many different coupled dynamical systems, depending on the balance between 
the internal dynamics of the individual systems and the coupling between them~\cite{Kur75,Pik01}. 
For the specific case of electric power grids, the individual systems are either generators or loads, and they are coupled to one 
another by power lines. Individual units have effective internal
dynamics determined by their nature -- they may be rotating machines, inertialess new renewable energy sources, 
load impedances and so forth -- and by the amount of power they generate or consume~\cite{Bialek08}. 
Rapid changes are currently affecting the structure of power grids which will no doubt impact their operating states~\cite{AnnualEnOutlook}. 
With higher penetration
levels of renewable energy sources, productions become decentralized, they have less inertia, they fluctuate more 
strongly~\cite{Bac13}, and are more vulnerable~\cite{Rohden2014}. It is expectable that
future power grids will be subjected more often to stronger external perturbations to which they may react more
strongly.

There is thus a clear need to better assess power grid vulnerability.
An assessment protocol has been advocated
in consensus and synchronization studies~\cite{Bamieh12,Summers15,Siami14,Fardad14},
which starts from a stable operating state, perturbs it and quantifies the magnitude of the induced transient
excursion through various performance measures. 
Focusing on Dirac-$\delta$, nodal perturbations -- instantaneous changes in generation or consumption --
performance measures have been proposed, which
can be formulated as the $\mathcal{L}_2$-norm of a performance output or 
the squared $\mathcal{H}_2$-norm of the input/output map.
The approach is mathematically elegant because these norms can be expressed in terms of 
observability Gramians~\cite{Zhou96}. 
Exported to electric power grids, performance measures evaluate additional  
transmission losses incurred during the transient as synchronous machines oscillate relative to one another~\cite{bamieh2013price, Gayme15, Gayme16}
or the primary control effort necessary to restore synchrony~\cite{Poola17}. 
Quite interestingly, \cite{Gayme16} relates additional  
transmission losses  to a graph theoretical distance metric known as the resistance distance~\cite{Klein93,Stephenson89}.
To the best of our knowledge, investigations of performance measures of synchronized states have considered
nodal perturbations only. In this manuscript, we extend these investigations
to line contingencies which are at least as important for evaluating the robustness of electric power grids. 

In the case of a line contingency, the disturbance acts on the network Laplacian matrix and is thus 
a multiplicative perturbation, a priori fundamentally different and harder to treat than the additive nodal perturbations
considered so far. 
For the case of a Dirac-$\delta$ perturbation, we show below how it can be recast as an effective, tractable additive perturbation. 
A second difficulty we overcome is that so far, analytical results for quadratic performance measures
have been obtained only for 
networks where static, passive nodes have been algebraically suppressed by Kron reduction~\cite{Dorfler13Kron}.
This is not a problem for nodal perturbations, however for the line contingencies of interest here a true line fault
in the physical network including passive nodes translates differently into the Kron-reduced network, 
depending on whether the faulted line connects two passive nodes, 
two active nodes, or one passive to one active node in the physical network. 
Below, we therefore map single-line faults in the physical network onto the Kron-reduced network.
We obtain results that differentiate between the three cases just mentioned. Finally, 
because our analytical approach is restricted to Dirac-$\delta$ perturbations, 
we compare our theory to numerical simulations with finite-time line faults. We confirm our analytical results
for faults lasting typically up to few AC cycles.  For uniform damping over inertia ratios, we find that the agreement between theory
and numerics is better for longer fault durations in systems with larger inertias  for which 
eigenmodes of the network are harder to excite.

References~\cite{bamieh2013price, Gayme15, Gayme16,Poola17} focus on performance measures 
which can be expressed as $\mathcal{H}_2$-norms for a state-space system. 
Generally, the observability Gramian required to evaluate 
$\mathcal{H}_2$-norms is defined implicitly by a Lyapunov equation which is typically solved numerically.
An additional contribution of our work is to derive an explicit solution of the Lyapunov equation,
under the assumption that synchronous machines have uniform damping over inertia ratios.
The Laplacian systems we consider have a marginally stable mode. It is well known that the presence of this mode may
lead to divergences in the observability Gramian, and the standard way to avoid these divergences is to 
chose the output orthogonal to that mode. Inspired by well established {\it regularization} procedures in theoretical physics~\cite{Itz80},
we construct an alternative approach which temporarily stabilizes that mode by adding a small, positive value $\epsilon$ 
to each eigenvalue of the Laplacian until the end of the calculation, at which point we send $\epsilon$ to zero. 
The procedure 
cancels removable singularities, furthermore it renders our explicit solution of the Lyapunov equation
mathematically rigorous in all cases as long as $\epsilon$ is finite. It
clarifies the origin of divergences, when they cancel out and when they do not. 
In the latter case the procedure is not successful, indicating that the problem is ill-posed. 
Because the procedure is
imported from a different field of research than the standard readership of this journal, we spend some time to describe it below. 

This paper is organized as follows. Section \ref{sec:Notations} introduces the mathematical notations
and defines the effective resistance distance.
The high voltage AC electric network model, and the observability Gramian formalism are outlined in
Section~\ref{sec:Model}. Section~\ref{sec:Observability Gramian expr} derives a closed-form expression for the 
observability Gramian in general cases.
The new application of the Gramian formalism to 
line contingencies is discussed in Section~\ref{sec:Line contingencies, formalism} and is supported by the numerical simulations presented in 
Section~\ref{sec:Numerics}. A brief conclusion is given in Section~\ref{sec:Conclusion}.

\section{General mathematical properties and notation}\label{sec:Notations}
Given the vector ${\bm v}\in\mathbb{R}^N$ and the matrix ${\bm M}\in \mathbb{R}^{N\times N}$
we denote their transpose by ${\bm v}^\top$ and ${\bm M}^\top$.
For any two vectors ${\bm u}, {\bm v}\in \mathbb{R}^N$, ${\bm u}{\bm v}^\top$ is the matrix in $\mathbb{R}^{N\times N}$
having as $i,j^\textrm{th}$ component the scalar $u_iv_j$ and
$\textrm{diag}(\{v_i\})$ denotes the diagonal matrix having $v_1,\ldots,v_N$ as diagonal entries.
Let ${\bm e}_l\in \mathbb{R}^{N}$ with $l\in\{1,\ldots,N\}$ denote the unit vector 
with components $({e}_l)_i = \delta_{il}$ with the Kronecker symbol $\delta_{il}=0$ if $i \ne l$, $\delta_{ii}=1$.
We define ${\bm e}_{(l,q)} = {\bm e}_l-{\bm e}_q \in \mathbb{R}^{N}$.

We denote undirected weighted graphs by $\mathcal{G}=(\mathcal{N},\mathcal{E}, \mathcal{W})$ where 
$\mathcal{N}$ is the set of its $N$ vertices,
$\mathcal{E}$ is the set of edges, and $\mathcal{W}=\{w_{ij}\}$ is the set of edge weights,
with $w_{ij}=0$ whenever $i$ and $j$ are not connected by an edge,
and $w_{ij}=w_{ji}>0$ otherwise.
The graph Laplacian ${\bm L}\in\mathbb{R}^{N\times N}$ is the symmetric matrix given by 
${\bm L} = \sum_{i<j}w_{ij}{\bm e}^{\phantom{\top}}_{(i,j)}{{\bm e}^\top_{(i,j)}}$.
We denote by $\{\lambda_1,\ldots\lambda_N\}$ and $\{{\bm u}^{(1)},\ldots,{\bm u}^{(N)}\}$ the eigenvalues and orthonormalized eigenvectors of $\bm L$.
The zero row and column sum property of $\bm L$ implies that $\lambda_1=0$ and that $({\bm u}^{(1)})^\top=(1,\ldots,1)/\sqrt{N}$.
All remaining eigenvalues of $\bm L$ are strictly positive in connected graphs, $\lambda_i>0$ for $i=2,\ldots,N$.
The orthogonal matrix $\bm T\in \mathbb{R}^{N\times N}$ having ${\bm u}^{(i)}$ as $i^\textrm{th}$ column diagonalizes $\bm L$,
i.e. ${\bm T}^\top {\bm L} {\bm T} = \textrm{diag}(\{\lambda_i\})$.
The Moore-Penrose pseudoinverse of ${\bm L}$ is given by ${\bm L}^\dagger = {\bm T}\textrm{diag}(\{0,\lambda_2^{-1},\ldots,\lambda_N^{-1}\}){\bm T}^\top$
and is such that ${\bm L}{\bm L}^\dagger={\bm L}^\dagger{\bm L}=\mathbb{I}-{\bm u}^{(1)}{{\bm u}^{(1)}}^\top$
with $\mathbb{I}\in\mathbb{R}^{N\times N}$ denoting the identity matrix.

The effective resistance distance between any two nodes $i$ and $j$ of the network is 
defined as~\cite{Klein93, Stephenson89} 
\begin{equation}
\Omega_{ij}={{\bm e}^\top_{(i,j)}} {\bm L}^\dagger{\bm e}^{\phantom{\top}}_{(i,j)} \, .
\end{equation} 
This quantity is a graph theoretical
distance metric satisfying the properties: i) $\Omega_{ii}=0~\forall i\in\mathcal{N}$, 
ii) $\Omega_{ij}\geq0~\forall i\neq j\in\mathcal{N}$,
and iii) $\Omega_{ij}\leq\Omega_{ik}+\Omega_{kj}~\forall i,j,k\in\mathcal{N}$.
It is known as the \textit{resistance} distance because if one replaces the edges of $\mathcal{G}$ by resistors
with a conductance $1/R_{ij}=w_{ij}$, then $\Omega_{ij}$ is equal to the
equivalent network resistance when a current is injected at node $i$ and extracted at node $j$ with no injection anywhere
else. The expression of the resistance distance in terms of the eigenvalues and eigenvectors of $\bm L$ is 
given by~\cite{klein1997graph, xiao2003resistance, coletta17}
\begin{equation}
\Omega_{ij}=\sum_{l\geq2}\lambda_l^{-1}(u_{i}^{(l)}-u_{j}^{(l)})^2 \, .
\end{equation}
\section{Power Network Model and quadratic performance measures}\label{sec:Model}
We consider the dynamics of high voltage transmission power networks in the DC power flow approximation.
This approximation of the full nonlinear dynamics assumes uniform and constant voltage magnitudes, 
purely susceptive transmission lines and small voltage angle differences between connected nodes. 
It is standardly justified in very high voltage transmission grids where line conductances are 
typically ten times smaller than line susceptances, and which operate at angle differences smaller than 
$\sim 30^o$, justifying a linearization of the power flow they carry, $\sin(\theta_i-\theta_j) \simeq \theta_i-\theta_j$.
Accordingly, 
the steady state power flow equations {\color{blue} linearly} relate the active power injections $\bm P$ 
to the voltage angles $\bm \theta$ at every node, $\bm P ={\bm L}_{\bm b} \bm \theta$. 
Here, ${\bm L}_{\bm b}$ is the Laplacian matrix of 
the electric network, with a subscript indicating that its edge weights are given by the susceptances of the transmission lines,
$w_{ij}=b_{ij}\geq0$.
We assume that each node of the network has a synchronous machine (generator or consumer) of 
rotational inertia $m_i>0$ and damping coefficient $d_i>0$.
We call such nodes {\it active nodes}.
For constant voltages, the dynamics is governed by the \textit{swing} equations \cite{Bialek08}.
In the frame rotating at the nominal frequency they read
\begin{equation}\label{eq:Swing}
{\bm M}\ddot{{\bm \theta}} = -{\bm D}\dot{{\bm \theta}} + {\bm P} - {{\bm L}_{\bm b}}{\bm \theta},
\end{equation}
with ${\bm M}=\textrm{diag}(\{m_i\})$ and $\bm D = \textrm{diag}(\{d_i\})$.
Subject to a power injection perturbation ${\bm P} \rightarrow {\bm P}+{\bm p}(t)$, 
the system deviates from the nominal operating point 
$({\bm \theta}^\star, {\bm\omega}):=({{\bm L}^{\dagger}_{\bm b}}{\bm P},0)$ 
according to ${\bm \theta}(t)={\bm\theta}^\star+{\bm \varphi}(t)$, and
${\bm \omega}(t)=\dot{\bm \varphi}(t)$, where ${\bm \varphi}(t)$ and 
${\bm \omega}(t)$ measure angle and angular frequency deviations from the nominal operating point.
Following \cite{paganini2017global}, we
use $\overline{{\bm \varphi}}={\bm M}^{1/2}{\bm \varphi}$ and
$\overline{{\bm \omega}}={\bm M}^{1/2}{\bm \omega}$ to rewrite \eqref{eq:Swing} as 
\begin{equation}\label{eq:Swing3}
      \bigg[\begin{IEEEeqnarraybox*}[][c]{,c,}
       \dot{\overline{\bm \varphi}} \\
       \dot{\overline{\bm \omega}}
      \end{IEEEeqnarraybox*}\bigg]=
       \underbrace{\bigg[\begin{IEEEeqnarraybox*}[][c]{,c/c,}
        0 & \mathbb{I} \\
        -{\bm M}^{-1/2}{\bm L}_{\bm b}{\bm M}^{-1/2} & -{\bm M}^{-1}{\bm D} 
      \end{IEEEeqnarraybox*}\bigg]}_{\bm A}
      \bigg[\begin{IEEEeqnarraybox*}[][c]{,c,}
        \overline{\bm \varphi} \\
        \overline{\bm \omega} 
      \end{IEEEeqnarraybox*}\bigg]
      +\bigg[\begin{IEEEeqnarraybox*}[][c]{,c,}
        0 \\
        {\bm M}^{-1/2}{\bm p}
      \end{IEEEeqnarraybox*}\bigg]\,,
\end{equation}
which symmetrizes the four blocks of the stability matrix ${\bm A}$ connecting angle and angular frequency deviations
to their derivatives in \eqref{eq:Swing3} and simplifies the eigenbasis decomposition of $\bm A$ described below.
Equation~\eqref{eq:Swing3} captures the transient dynamics resulting from the 
perturbation ${\bm p}(t)$.
For asymptotically stable systems and perturbations that are short and weak enough that they leave the dynamics inside the basin of attraction 
of $\bm \theta^\star$, the operating point will eventually return to $(\bm \varphi,\bm \omega)=(0,0)$.

We want to characterize the transient by evaluating quadratic performance measures generically given by
\begin{equation}\label{eq:objective function}
\mathcal{P}= \int_{0}^{\infty}
 \big[\begin{IEEEeqnarraybox*}[][c]{,c,}
        \bm \varphi^\top
        \bm \omega^\top 
      \end{IEEEeqnarraybox*}\big]{\bm Q}
 \bigg[\begin{IEEEeqnarraybox*}[][c]{,c,}
        \bm \varphi \\
        \bm \omega 
      \end{IEEEeqnarraybox*}\bigg]\dd t\,, \;
      \bm Q=\bigg[\begin{IEEEeqnarraybox*}[][c]{,c/c,}
        {\bm Q}^{(1,1)} & 0 \\
        0 & {\bm Q}^{(2,2)} 
      \end{IEEEeqnarraybox*}\bigg] \, , 
\end{equation}
where the precise form of the symmetric matrix 
$\bm Q\in\mathbb{R}^{2N\times 2N}$ depends on the specific performance measure to investigate and
we assumed ${\bm p}(t<0)=0$. For Dirac-$\delta$ perturbations ${\bm p}(t)=\bm p_0 \, \delta(t)$
and initial conditions $(\bm \varphi(0),\bm \omega(0))=(0,0)$, \eqref{eq:Swing3} is explicitly solved yielding
\begin{equation}\label{eq: Definition B}
  \bigg[\begin{IEEEeqnarraybox*}[][c]{,c,}
        \overline{\bm \varphi}(t) \\
        \overline{\bm \omega}(t) 
      \end{IEEEeqnarraybox*}\bigg]=e^{{\bm A} t} \bigg[\begin{IEEEeqnarraybox*}[][c]{,c,}
        0 \\
        {\bm M}^{-1/2}{\bm p_0}
        \end{IEEEeqnarraybox*}\bigg] =: e^{{\bm A} t} \, {\bm B} \,.
\end{equation}
The performance measure $\mathcal{P}$ \eqref{eq:objective function} can be expressed as 
\begin{equation}\label{eq:objective function 2}
\mathcal{P}={\bm B}^\top \bm X {\bm B}\,,
\end{equation}
with the observability Gramian $ {\bm X} = \int_{0}^{\infty}e^{{\bm A}^\top t} \bm Q^\textrm{M} e^{{\bm A}t} \dd t$,
and
\begin{equation}
 {\bm Q}^\textrm{M} = \bigg[\begin{IEEEeqnarraybox*}[][c]{,c/c,}
        {\bm M}^{-1/2}{\bm Q}^{(1,1)}{\bm M}^{-1/2} & 0 \\
        0 & {\bm M}^{-1/2}{\bm Q}^{(2,2)}{\bm M}^{-1/2} 
      \end{IEEEeqnarraybox*}\bigg]\,.
\end{equation}

When the system \eqref{eq:Swing3} is asymptotically stable, the observability Gramian $\bm X$ satisfies
the Lyapunov equation
\begin{equation}\label{eq:Lyapunov equation}
{\bm A}^\top {\bm X} + {\bm X} {\bm A} = -{\bm Q}^\textrm{M}\,,
\end{equation}
The matrix $\bm A$ defined in \eqref{eq:Swing3} depends on the 
Laplacian matrix ${\bm L}_{\bm b}$, which satisfies $\sum_i ({\bm L}_{\bm b})_{ij}=\sum_j ({\bm L}_{\bm b})_{ij}=0$.
Accordingly, ${\bm A}$ has a marginally stable mode, ${\bm A}[{\bm M}^{1/2}{{\bm u}^{(1)}}, 0]^\top=0$.
The standard approach to deal with this mode is to consider performance measures
such that ${\bm u}^{(1)}\in\textrm{ker}(\bm Q^{(1,1)})$, in which case the observability Gramian is 
well defined by \eqref{eq:Lyapunov equation} with the additional constraint ${\bm X}[{\bm M}^{1/2}{{\bm u}^{(1)}}, 0]^\top=0$
\cite{bamieh2013price, Gayme16, Poola17}. 
Alternatively, one may subtract any homogeneous
angle shift $\phi_i(t) \rightarrow \phi_i(t) - \langle \phi(t) \rangle$ with $\langle \phi(t) \rangle = n^{-1} \sum_i \langle \phi(t) \rangle$
via a transformation of coordinates~\cite{Wu16}.
Here we instead construct a new mathematical approach for calculating quadratic performance 
measures of the form given in \eqref{eq:objective function}. Its construction will be detailed below, and  we briefly 
summarize it. Our formalism starts with a
formal stabilization of the marginally stable mode of $\bm A$
by adding an identical shift, parametrized by $\epsilon > 0$, to all eigenvalues of the Laplacian,
${\bm L}_{\bm b}\rightarrow {\bm L}_{\bm b}+\epsilon\mathbb{I}$, without changing 
the corresponding eigenvectors. In particular, this shift renders
asymptotically stable the marginally stable mode, but does not change its structure. 
At all intermediate calculational steps, the observability Gramian can be calculated from the Lyapunov 
equation \eqref{eq:Lyapunov equation} without further constraint, which in particular 
gives closed-form expressions from which the observability Gramian $\bm X$ is easily computed.
We remove the regularizing term $\epsilon \mathbb{I}$
at the end of the calculation and the procedure, common in theoretical physics~\cite{Itz80}, 
is successful only if the final result is finite. Among others, we will see below that this procedure
rigorously treats cases where the Gramian has removable singularities, as will be illustrated below. 

\begin{proposition}\label{prop:Regularized inverse of L}
The Laplacian ${\mathbb L}_{\bm b}(\epsilon) = {\bm L}_{\bm b}+\epsilon\mathbb{I}$
with $\epsilon>0$ is non singular and its inverse is
\begin{equation}\label{eq:Regularized inverse of L}
{\mathbb L}^{-1}_{\bm b}(\epsilon)={\bm T}{\rm diag}(\{\epsilon^{-1},(\lambda_2+\epsilon)^{-1},\ldots,(\lambda_N+\epsilon)^{-1}\}){\bm T}^\top\,, 
\end{equation}
where $\lambda_i$'s and ${\bm T}$ are the eigenvalues and the orthogonal matrix diagonalizing ${\bm L}_{b}$.
\end{proposition}
\begin{proof}
Because the identity matrix $\mathbb I$ commutes with ${\bm L}_{\bm b}$,
${\mathbb L}_{\bm b}(\epsilon)$ has the same eigenvectors as 
${\bm L}_{\bm b}$ with shifted eigenvalues $\lambda_i\rightarrow\lambda_i+\epsilon$.
Because $\epsilon>0$,  all eigenvalues are strictly positive and ${\mathbb L}_{\bm b}(\epsilon)$ is non singular. Since the eigenvectors 
are unchanged,  ${\mathbb L}_{\bm b}(\epsilon)$ is diagonalized by the same matrix as ${\bm L}_{\bm b}$
and one obtains \eqref{eq:Regularized inverse of L} for its inverse.
\end{proof}

\begin{remark}\label{remark on inverses}
From \eqref{eq:Regularized inverse of L} we see that $\lim_{\epsilon \rightarrow 0} \bm v^\top \, {\mathbb L}^{-1}_{\bm b}(\epsilon) \,
\bm v = \bm v^\top \bm L_{\bm b}^\dagger \bm v$ if $\bm v$ is orthogonal to the marginally stable mode $\bm u^{(1)}$. We will 
make extensive use of that property below. 
\end{remark}

\begin{proposition}\label{prop:Marginal mode assumption}
 Under the transformation ${\bm L}_{\bm b}\rightarrow {\mathbb L}_{\bm b}(\epsilon) = {\bm L}_{\bm b}+\epsilon\mathbb{I}$ with $\epsilon>0$, the system
 defined by \eqref{eq:Swing3} is asymptotically stable and has no marginally stable mode. 
\end{proposition}
\begin{proof}
 In Proposition~\ref{prop:Regularized inverse of L}, we 
 showed that for $\epsilon>0$, ${\mathbb L}_{\bm b}(\epsilon)$ has  strictly positive eigenvalues. 
 Under this condition \cite{Manik2014} and \cite{coletta2016linear}
 showed that all eigenvalues of $\bm A$ have a strictly negative real part. Therefore the system
 defined by \eqref{eq:Swing3} is asymptotically stable.
\end{proof}

Under the transformation ${\bm L}_{\bm b}\rightarrow {\mathbb L}_{\bm b}(\epsilon) ={\bm L}_{\bm b}+\epsilon\mathbb{I}$, \eqref{eq:Lyapunov equation} defines the observability Gramian with no additional constraint for $\epsilon >0$. 
We can then calculate $\bm X$ for any performance measure. We
take the limit $\epsilon\rightarrow0$ only at the end of the calculation. In the absence of any singularity, the procedure 
reproduces the results of standard approaches when the latter work. 
Additionally, our regularization procedure takes care of removable singularities when present, and therefore is able to treat cases
beyond the reach of the standard approaches.
This will be illustrated below.

Defining a system output
${\bm y}(t)=\sqrt{{\bm Q}^\textrm{M}}[\overline{\bm \varphi}(t),\overline{\bm \omega}(t)]^\top$,
\eqref{eq:Swing3} and \eqref{eq: Definition B} together with ${\bm y}(t)$ define an input/output system which we denote 
by $G=(\bm A,\bm B, \sqrt{{\bm Q}^\textrm{M}})$. 
For a specific perturbation, the quadratic performance measure $\mathcal{P}$ of \eqref{eq:objective function}
is the $\mathcal{L}_2$-norm $\mathcal{P}=\int_0^\infty {\bm y}^\top(t){\bm y}(t)\dd t\equiv\norm{\bm y}_{\mathcal{L}_2}$.
Furthermore, the squared $\mathcal{H}_2$-norm 
$\norm{G}_{\mathcal{H}_2}^2$ measures the sum of the system's responses to Dirac-$\delta$ impulses at every node,
$\norm{G}_{\mathcal{H}_2}^2=\sum_{i}\mathcal{P}^{(i)}$,
where $\mathcal{P}^{(i)}$ is the $\mathcal{L}_2$-norm of the system's output for a single Dirac-$\delta$ impulse at node $i$,
${\bm p}^{(i)}(t)=\delta(t){\bm e}_i$. 

In this formalism, the difficult step is to solve the Lyapunov equation
for $\bm X$. Despite some specific choices of $\bm Q$ 
for which analytical solutions have been found~\cite{Gayme15,Gayme16,Poola17}, this task is generally 
performed numerically.
For generic performance measures, there is no formal solution to the Lyapunov equation.
Section \ref{sec:Observability Gramian expr} fills this gap by explicitly deriving an analytical expression for the observability Gramian in terms of the eigenvectors
of ${\bm M}^{-1/2}\bm L_{\bm b}{\bm M}^{-1/2}$ for additive perturbations. 

\section{Closed form expression for the Observability Gramian}\label{sec:Observability Gramian expr}
\begin{proposition}\label{prop:Sol Lyapunov eq}
 Let $\bm A$ be a non symmetric, diagonalizable matrix with eigenvalues ${\rm Re} (\mu_i) < 0~\forall i$. Let 
 ${\bm T}_R$ (${\bm T}_L$) denote the matrix whose columns (rows) are the right (left) eigenvectors of $\bm A$.
 The observability Gramian $\bm X$, solution of the Lyapunov equation~\eqref{eq:Lyapunov equation} is given by 
 \begin{equation}\label{eq:Solution of lyapunov equation}
 X_{ij} = \sum_{l,q}\frac{-1}{\mu_l+\mu_q}(T_L)_{li}(T_L)_{qj}\left( {\bm T}_R^\top \bm Q^\textrm{M} {\bm T}_R\right)_{lq}\,.
\end{equation}
\end{proposition}

\begin{proof}
By definition, ${\bm T}_L$ and ${\bm T}_R$ fulfill the bi-orthogonality condition ${\bm T}_L{\bm T}_R={\bm T}_R{\bm T}_L=\mathbb{I}$, 
and ${\bm T}_L {\bm A} {\bm T}_R = \bm \mu$
with $\bm \mu =\textrm{diag}(\{\mu_i\})$. Using this transformation in \eqref{eq:Lyapunov equation} one has
\begin{equation}
 {\bm \mu} \overline{\bm X} + \overline{\bm X} {\bm \mu} = -{\bm T}^\top_R{\bm Q}^\textrm{M}{\bm T}_R\,, \quad \overline{\bm X} = {\bm T}^\top_R {\bm X} {\bm T}_R\,,
\end{equation}
which yields
\begin{equation}\label{eq:pre-solution of lyapunov equation}
 \overline{X}_{lq} = \frac{-1}{\mu_l+\mu_q}({\bm T}^\top_R{\bm Q}^\textrm{M}{\bm T}_R)_{lq}\, ,
\end{equation}
which is always finite under the assumption that ${\color{blue}  {\rm Re} (\mu_i) < 0~\forall i}$.
Finally, using $\bm X = {\bm T}_L^\top\overline{\bm X}{\bm T}_L$ one obtains \eqref{eq:Solution of lyapunov equation}.
\end{proof}

\begin{remark}\label{remark on def of the obs Gramian}
Since Proposition \ref{prop:Sol Lyapunov eq} holds for ${\rm Re} (\mu_i) < 0~\forall i$, 
it is in particular satisfied by the matrix $\bm A$ defined in \eqref{eq:Swing3} with ${\bm L}_{\bm b} \rightarrow {\mathbb L}_{\bm b}(\epsilon) = {\bm L}_{\bm b}+\epsilon\mathbb{I}$ (see proof of Proposition~\ref{prop:Marginal mode assumption}).
The limit $\epsilon \rightarrow 0$ cannot be directly taken 
in \eqref{eq:Solution of lyapunov equation} and \eqref{eq:pre-solution of lyapunov equation} because the 
$l=q=1$ term has a vanishing denominator. 
Instead, we keep $\epsilon > 0$ to calculate performance measures. We
take the limit $\epsilon\rightarrow0$ only at the very end of the calculation. We will see that the limit 
is well defined for some, but not all choices of a performance measure and of a perturbation. 
\end{remark}

Next we relate the explicit expression \eqref{eq:Solution of lyapunov equation} of the observability Gramian 
to the eigenvectors of ${\bm M}^{-1/2}\mathbb L_{\bm b}(\epsilon){\bm M}^{-1/2}$.

\begin{assumption}\label{ass:Homogeneous inertia damping ratio}
 All synchronous machines have uniform damping over inertia ratios $d_i/m_i= \gamma>0~\forall i$.
 This is a standard assumption when analytically calculating performance measures
 such as those in \eqref{eq:objective function}~\cite{Bamieh12,bamieh2013price,Gayme15,Gayme16,Poola17,paganini2017global}. 
 Machine measurements indicate that the ratio $\gamma=d_i/m_i$ varies by about an order of magnitude from rotating machine to rotating 
 machine \cite{KouMachineReport}, so that this assumption, while strictly speaking not justified, is not  unrealistic. 
With this assumption, the matrix $\bm A$ in \eqref{eq:Swing3} can be rewritten as
 \begin{equation}
{\bm A}  \rightarrow 
 \bigg[\begin{IEEEeqnarraybox*}[][c]{,c/c,}
        0 & \mathbb{I} \\
        -{\bm M}^{-1/2}{\bm L}_{\bm b}{\bm M}^{-1/2} & -\gamma \mathbb{I} 
      \end{IEEEeqnarraybox*}\bigg] \, .
 \end{equation}
Following this, $\bm A$ is easily block-diagonalized in the eigenbasis of ${\bm M}^{-1/2}{\bm L}_{\bm b}{\bm M}^{-1/2}$
and its eigenvalues straightforwardly calculated as we next proceed to show. A perturbation theory valid for weak 
inhomogeneities in $d_i/m_i$ is currently under construction~\cite{Pag18}. We are unaware of any analytical treatement of the 
general case with arbitrary $d_i/m_i$.
\end{assumption}

\begin{proposition}\label{prop:Trasformation}
Consider the power system model defined in \eqref{eq:Swing3} 
with ${\bm L}_{\bm b} \rightarrow {\mathbb L}_{\bm b}(\epsilon) = {\bm L}_{\bm b}+\epsilon\mathbb{I}$, $\epsilon >0$ and
under Assumption \ref{ass:Homogeneous inertia damping ratio}.
The left and right transformation matrices 
 $\bm T_L$ and $\bm T_R$ diagonalizing ${\bm A}$,  
 \begin{equation}
{\bm T}_L {\bm A}{\bm T}_R = \bigg[\begin{IEEEeqnarraybox*}[][c]{,c/c,}
        {\rm diag}(\{ \mu_j^+\}) & 0 \\
        0 & {\rm diag}(\{ \mu_j^-\})
    \end{IEEEeqnarraybox*}\bigg]\,,
\end{equation}
with ${\bm T}_L {\bm T}_R=\mathbb{I}$,
are related to the matrix ${\bm T}_\textrm{M}$ whose columns are
 the eigenvectors of ${\bm M}^{-1/2}\mathbb L_{\bm b}(\epsilon){\bm M}^{-1/2}$ 
 through
\begin{equation}\label{eq:T_R}
 {\bm T}_R = \bigg[\begin{IEEEeqnarraybox*}[][c]{,c/c,}
        {\bm T}_\textrm{M} & 0 \\
        0 & {\bm T}_\textrm{M}
    \end{IEEEeqnarraybox*}\bigg]
    \bigg[\begin{IEEEeqnarraybox*}[][c]{,c/c,}
        {\rm diag}(\{1/\sqrt{\Gamma_j}\}) & {\rm diag}(\{\text{i}/\sqrt{\Gamma_j}\}) \\
        {\rm diag}(\{\mu_j^+/\sqrt{\Gamma_j}\}) & {\rm diag}(\{\text{i}\mu_j^-/\sqrt{\Gamma_j}\})
    \end{IEEEeqnarraybox*}\bigg]\,,\qquad
\end{equation}
and
\begin{equation}\label{eq:T_L}
 {\bm T}_L = 
 \bigg[\begin{IEEEeqnarraybox*}[][c]{,c/c,}
        {\rm diag}(\{-\mu_j^-/\sqrt{\Gamma_j}\}) & {\rm diag}(\{1/\sqrt{\Gamma_j}\}) \\
        {\rm diag}(\{-\text{i}\mu_j^+/\sqrt{\Gamma_j}\}) & {\rm diag}(\{\text{i}/\sqrt{\Gamma_j}\})
    \end{IEEEeqnarraybox*}\bigg]
 \bigg[\begin{IEEEeqnarraybox*}[][c]{,c/c,}
        {\bm T}^\top_\textrm{M} & 0 \\
        0 & {\bm T}^\top_\textrm{M}
    \end{IEEEeqnarraybox*}\bigg]\,,\qquad
\end{equation}
with 
\begin{equation}\label{eq:Eigenvalues A}
 \mu_j^{\pm}=\frac{1}{2}\left(-\gamma\pm\Gamma_j\right)\,,  \Gamma_j=\sqrt{\gamma^2-4\lambda_j^\textrm{M}}\,, 
  j=1, \ldots N \,.
\end{equation}
\end{proposition}
\begin{proof}
Under Assumption \ref{ass:Homogeneous inertia damping ratio}, 
one has ${\bm M}^{-1}\bm D=\gamma\mathbb{I}$. 
Thus ${\bm M}^{-1}\bm D$ and ${\bm M}^{-1/2}\mathbb L_{\bm b}(\epsilon){\bm M}^{-1/2}$
have a common basis of eigenvectors $\forall \epsilon >0$.
Since ${\bm M}^{-1/2}\mathbb L_{\bm b}(\epsilon){\bm M}^{-1/2}$ is symmetric, it has a real spectrum with eigenvalues
$\lambda^\textrm{M}_i=\lambda_i^\textrm{M}(\epsilon)$ with
\begin{equation}\label{eq:lambdamtm}
{\bm T}_\textrm{M}^\top\left({\bm M}^{-1/2}\mathbb L_{\bm b}(\epsilon){\bm M}^{-1/2}\right){\bm T}_\textrm{M} = 
{\bm \Lambda}_\textrm{M} :=\textrm{diag}(\{\lambda_i^\textrm{M}\})\, .
\end{equation}
The $\epsilon$-dependance of $\lambda_i^\textrm{M}(\epsilon)$ is not trivial, however it can
be expanded in a power series in $\epsilon$,
which converges and has controlled accuracy at small enough $\epsilon$~\cite{Stewart}.
The transformation
\begin{equation}
\bigg[\begin{IEEEeqnarraybox*}[][c]{,c/c,}
        {\bm T}_\textrm{M}^\top & 0 \\
        0 & {\bm T}_\textrm{M}^\top 
    \end{IEEEeqnarraybox*}\bigg] 
{\bm A}
\bigg[\begin{IEEEeqnarraybox*}[][c]{,c/c,}
        {\bm T}_\textrm{M} & 0 \\
        0 & {\bm T}_\textrm{M}
    \end{IEEEeqnarraybox*}\bigg]
    =
\bigg[\begin{IEEEeqnarraybox*}[][c]{,c/c,}
        0 & \mathbb{I} \\
        -{\bm \Lambda}_\textrm{M} & -\gamma\mathbb{I}
    \end{IEEEeqnarraybox*}\bigg]
\end{equation}
leads after index reordering to a matrix with a block diagonal structure composed of $2\times2$ blocks of the form
\begin{equation}\label{eq:2x2 block}
\bigg[\begin{IEEEeqnarraybox*}[][c]{,c/c,}
        0 & 1 \\
        -\lambda_i^\textrm{M} & -\gamma
\end{IEEEeqnarraybox*}\bigg]\,.
\end{equation}
Diagonalizing the $2\times2$ blocks in \eqref{eq:2x2 block}, one obtains the eigenvalues of $\bm A$ in~\eqref{eq:Eigenvalues A}.
For $\Gamma_i\neq0$, we use the right and left eigenvectors of \eqref{eq:2x2 block} to obtain 
the full transformation which diagonalizes $\bm A$. 
\end{proof}

Equations \eqref{eq:T_R} and \eqref{eq:T_L} relate the eigenvectors of $\bm A$ to those of ${\bm M}^{-1/2}\mathbb L_{\bm b}(\epsilon){\bm M}^{-1/2}$, following which one can then express $X_{ij}$ in \eqref{eq:Solution of lyapunov equation} in terms of the 
eigenvectors of ${\bm M}^{-1/2}\mathbb L_{\bm b}(\epsilon){\bm M}^{-1/2}$ under the condition that $\Gamma_i\neq0$, $\forall i$.
Proposition~\ref{prop:Trasformation} is therefore valid for all values of $\gamma$ except a set of zero
measure. 

The linearity of the Lyapunov equation \eqref{eq:Lyapunov equation} with respect to both $\bm X$ and $\bm Q^\textrm{M}$ 
implies that for performance measures involving both frequency and voltage angle degrees of freedom, 
the observability Gramian is a linear combination of an observability Gramian for  
a purely angle-dependent and another for a purely frequency-dependent performance measure.
Without loss of generality, we therefore address separately classes of performance measures involving 
frequency degrees of freedom only and those involving angle degrees of freedom only.
We focus on the ${\bm X}^{(2,2)}$ block of the observability Gramian because, 
from \eqref{eq: Definition B} and \eqref{eq:objective function 2} it is the only block of ${\bm X}$ that appears in 
$\mathcal{L}_2$- and squared $\mathcal{H}_2$-norms.
\begin{proposition}[\bf{Observability Gramian for frequency-based performance measures}]\label{prop: Obs Frequency}
 Consider the power system model defined in \eqref{eq:Swing3} with ${\bm L}_{\bm b} \rightarrow {\mathbb L}_{\bm b}(\epsilon) = {\bm L}_{\bm b}+\epsilon\mathbb{I}$.
 Under Assumption \ref{ass:Homogeneous inertia damping ratio}, the $\bm X^{(2,2)}$ block of the observability Gramian associated with the quadratic performance measure defined in 
 \eqref{eq:objective function} with ${\bm Q}^{(1,1)}=0$, and ${\bm Q}^{(2,2)}\neq0$ is given by
\begin{IEEEeqnarray}{lll}\label{eq:X22 frequency}
\nonumber
 X^{(2,2)}_{ij}&=&\sum_{l,q=1}^N(T_\textnormal{M})_{il}(T_\textnormal{M}^\top)_{qj}({\bm T}_\textnormal{M}^\top {\bm M}^{-1/2}{\bm Q}^{(2,2)}{\bm M}^{-1/2}{\bm T}_\textnormal{M})_{lq}\\ 
 &\times&\left[\frac{\gamma(\lambda_l^\textnormal{M}+\lambda_q^\textnormal{M})}{2\gamma^2(\lambda_l^\textnormal{M}+\lambda_q^\textnormal{M})+(\lambda_q^\textnormal{M}-\lambda_l^\textnormal{M})^2}\right]\,,
\end{IEEEeqnarray}
 where $\lambda_l^\textnormal{M}=\lambda_l^\textnormal{M}(\epsilon)$ and $\bm T_\textnormal{M}$ are defined in 
 \eqref{eq:lambdamtm}.
\end{proposition}

\begin{proof} 
Inserting \eqref{eq:T_R} and \eqref{eq:T_L} into \eqref{eq:Solution of lyapunov equation}, 
under the assumption that ${\bm Q}^{(1,1)}=0$ and taking the indices $i\,,j\in\{N+1,\ldots,2N\}$ 
to access ${\bm X}^{(2,2)}$ yields
\begin{IEEEeqnarray}{lll}
\nonumber
 X^{(2,2)}_{ij}&=&\sum_{\substack{l,q=1 \\ n,p=1}}^N
 \frac{(T^\top_\textrm{M})_{li}(T^\top_\textrm{M})_{qj}(T_\textrm{M})_{nl}(T_\textrm{M})_{pq}Q^{(2,2)}_{np}}{\sqrt{m_n m_p}~\Gamma_l\Gamma_q} \\ \nonumber
 &\times& \left[\frac{\mu_q^+\mu_l^-}{\mu_l^- + \mu_q^+}+\frac{\mu_q^-\mu_l^+}{\mu_l^+ + \mu_q^-}-\frac{\mu_q^+\mu_l^+}{\mu_l^+ + \mu_q^+}-\frac{\mu_q^-\mu_l^-}{\mu_l^- + \mu_q^-}\right]\,,
\end{IEEEeqnarray}
which simplifies to \eqref{eq:X22 frequency} using \eqref{eq:Eigenvalues A} and a little bit of algebra.
\end{proof}

\begin{proposition}[\bf{Observability Gramian for angle-based performance measures}]\label{prop: Obs angle}
Consider the power system model defined in \eqref{eq:Swing3} with ${\bm L}_{\bm b} \rightarrow {\mathbb L}_{\bm b}(\epsilon) = {\bm L}_{\bm b}+\epsilon\mathbb{I}$.
Under Assumption \ref{ass:Homogeneous inertia damping ratio}, the $\bm X^{(2,2)}$ block of the observability Gramian associated with the quadratic performance measure defined in 
 \eqref{eq:objective function} with ${\bm Q}^{(1,1)}\neq0$, and ${\bm Q}^{(2,2)}=0$ is given by
\begin{IEEEeqnarray}{lll}\label{eq:X22 angle}
\nonumber
 X^{(2,2)}_{ij}&=&\sum_{l,q=1}^N(T_\textnormal{M})_{il}(T_\textnormal{M}^\top)_{qj}({\bm T}_\textnormal{M}^\top {\bm M}^{-1/2}{\bm Q}^{(1,1)}{\bm M}^{-1/2}{\bm T}_\textnormal{M})_{lq}\\ 
 &\times&\left[\frac{2\gamma}{2\gamma^2(\lambda_l^\textnormal{M}+\lambda_q^\textnormal{M})+(\lambda_q^\textnormal{M}-\lambda_l^\textnormal{M})^2}\right]\,,
\end{IEEEeqnarray}
 where $\lambda_l^\textnormal{M}=\lambda_l^\textnormal{M}(\epsilon)$ and $\bm T_\textnormal{M}$ are defined in 
 \eqref{eq:lambdamtm}.
\end{proposition}
\begin{proof}
Inserting \eqref{eq:T_R} and \eqref{eq:T_L} into \eqref{eq:Solution of lyapunov equation},
under the assumption that ${\bm Q}^{(2,2)}=0$ and taking the indices $i\,,j\in\{N+1,\ldots,2N\}$ 
to access ${\bm X}^{(2,2)}$ yields
\begin{IEEEeqnarray}{lll}
\nonumber
X^{(2,2)}_{ij}&=&\sum_{\substack{l,q=1\\n,p=1}}^N\frac{(T_\textrm{M}^\top)_{li}(T_\textrm{M}^\top)_{qj}(T_\textrm{M})_{nl}(T_\textrm{M})_{pq}Q^{(1,1)}_{np}}{\sqrt{m_n m_p}~\Gamma_l\Gamma_q} \\ \nonumber
 &\times& \left[\frac{1}{\mu_l^- + \mu_q^+}+\frac{1}{\mu_l^+ + \mu_q^-}-\frac{1}{\mu_l^+ + \mu_q^+}-\frac{1}{\mu_l^- + \mu_q^-}\right]\,,
\end{IEEEeqnarray}
which simplifies to \eqref{eq:X22 angle} using \eqref{eq:Eigenvalues A} and a little bit of algebra.
\end{proof}

\begin{remark}\label{Remark on expressions for the obs gramian}
According to Proposition \ref{prop:Marginal mode assumption} and \ref{prop:Trasformation}, 
$\lambda_i^\textnormal{M}\equiv\lambda_i^\textnormal{M}(\epsilon)$,
and $\lambda_1^\textnormal{M}(0)=0$. Therefore, potentially diverging terms 
in the limit $\epsilon\rightarrow0$ in \eqref{eq:X22 frequency} and \eqref{eq:X22 angle} are those with $l=q=1$,
because the corresponding factor inside the square bracket has a vanishing denominator.
For the frequency case, \eqref{eq:X22 frequency} is well behaved when $\epsilon\rightarrow0$,
because for $l=q$, the square bracket in \eqref{eq:X22 frequency} goes to $1/2 \gamma$ regardless of $\epsilon$. 
For the angle performance measure \eqref{eq:X22 angle} the square bracket for $l=q=1$ has a pole of order one, however. 
This is so, because generally $[\bm u^{(1)},0]^\top$ is not in the kernel of angle performance measure $\bm{Q}^{(1,1)} \ne 0$, $\bm{Q}^{(2,2)}=0$. 

It is interesting to explore how divergences occur and we do it here for the case of uniform inertia, 
$m_i=m~\forall i$, in which case ${\bm T}_\textrm{M}\equiv{\bm T}$ (see Section~\ref{sec:Notations}) 
and $\lambda_i^\textrm{M}=(\lambda_i+\epsilon)/m$, and
for a single-node power injection pulse at node $s$, ${\bm p}(t)=\delta(t){\bm e}_s$ 
and $\bm B=[0,m^{-1/2}{\bm e}_s]^\top$.
The observability Gramian \eqref{eq:X22 angle} for $\bm{Q}^{(1,1)} = \mathbb I$ becomes
\begin{IEEEeqnarray}{c}\label{eq:Gramian coherence}
\nonumber 
X^{(2,2)}_{ij}=\frac{1}{2\gamma}\sum_{l=1}^N(\lambda_l+\epsilon)^{-1}T_{il}(T^\top)_{lj}\\
\label{eq:x22} \Rightarrow {\bm X}^{(2,2)}=\frac{1}{2\gamma}{\mathbb L}^{-1}_{\bm b}(\epsilon)\,.\qquad
\end{IEEEeqnarray}
For this single-node pulse, we therefore obtain
\begin{equation}\label{eq:spulse}
 \mathcal{P}^{(s)} = \frac{1}{2d}\bigg[{\epsilon}^{-1}{u_s^{(1)}}^2+\sum_{l\geq2}^N(\lambda_l+\epsilon)^{-1}{u_s^{(l)}}^2\bigg]\,,
\end{equation}
which diverges when $\epsilon \rightarrow 0$ because of the first term in the square bracket, which depends 
on the marginally stable mode of the Laplacian. 
If instead one considers two simultaneous pulses of opposite magnitude at nodes $s$ and $s'$, then 
${\bm p}(t)=\delta(t)({\bm e}_s-{\bm e}_{s'})$ and one obtains instead of \eqref{eq:spulse}
\begin{IEEEeqnarray}{lll}\label{eq:Pst}
\mathcal{P}^{(s-s')} = 
\mathcal{P}^{(s)}-\mathcal{P}^{(s')} &= &\frac{1}{2d}
\sum_{l\geq2}^N(\lambda_l+\epsilon)^{-1}({u_s^{(l)}}^2-{u_{s'}^{(l)}}^2) \, ,\qquad
\end{IEEEeqnarray}
because the zero mode has constant components, $u_{s}^{(1)}=u_{s'}^{(1)}$.
The $l=1$-term has disappeared, even though the marginally stable mode of $\bm L_{\bm b}$ is not in the 
kernel of $\bm{Q}^{(1,1)}=\mathbb{I}$, because $m^{-1/2}({\bm e}_s-{\bm e}_{s'})$ is orthogonal to
$\bm u^{(1)}$. Accordingly, there is no divergence in the 
performance measure ${\cal P} = {\bm B}^\top \bm X {\bm B}$ as $\epsilon \rightarrow 0$ in this case. 

Further using $u_{i}^{(1)}=1/\sqrt{N}~\forall i$ and the definition of the resistance distance given at the end of 
Section~\ref{sec:Notations}, \eqref{eq:Pst} finally gives
\begin{IEEEeqnarray}{lll}\label{eq:Difference 2 norms}
\mathcal{P}^{(s-s')}  &=& \frac{1}{2dN}\bigg[ \sum_i \Omega_{s i} -\sum_i \Omega_{s' i}\bigg]\,,
\end{IEEEeqnarray}
The quantity $N^{-1}\sum_i\Omega_{si}$ is the average resistance distance separating
node $s$ from the rest of the network. Its inverse is known as the closeness centrality \cite{Bozzo13}.
One concludes that the response $\mathcal{P}^{(s-s')}$ 
to two simultaneous pulses of opposite magnitude is positive (negative) 
if the centrality of node $s'$ is greater (smaller) than that of node $s$ -- 
power fluctuations occurring at nodes which are more central have a smaller impact on transient voltage angle 
fluctuations. 

\end{remark}

It is unclear to us how to derive \eqref{eq:Difference 2 norms} without using the
regularization scheme described above. This result enables to rank nodes according to the network response to 
a local perturbations on each of them. 
The result is finite because the divergence in  $\mathcal{P}^{(s)}$ is the same for all nodes $s$, as 
it originates from homogeneous rotations of all voltage angles, $\varphi_i \rightarrow \varphi_i + \Delta_0$, $\forall i$.
We note that similar expressions as \eqref{eq:X22 frequency} and \eqref{eq:X22 angle}
were recently derived in \cite{paganini2017global} working in the Laplace frequency domain.

\section{Performance measures under line contingencies: Formalism}\label{sec:Line contingencies, formalism}

The results presented so far rely on the assumption that all nodes of the network are governed by the swing
dynamics of \eqref{eq:Swing}, i.e. all nodes are active, thus have synchronous machines with inertia.
In real power networks however, some so-called \textit{passive} nodes are static, i.e. they have no dynamics.
Kron reduction eliminates passive nodes and formulates the dynamics of the network in terms of a 
swing equation similar to \eqref{eq:Swing}, involving only the voltage angles of the synchronous machines, 
all with a finite inertia, on an effective, reduced network~\cite{Dorfler13Kron}.
In Secs. \ref{sec:Line contingencies, formalism} and \ref{sec:Numerics} we therefore distinguish between Laplacians of the 
physical and
Kron reduced networks denoted $\bm L^\textrm{ph}_{\bm b}$ and $\bm L_\textrm{red}$ respectively.
Inertialess nodes are not always passive, however, and usually have a dynamics of their own~\cite{Berg81}. 
In our approach, we Kron-reduce all the inertialess nodes, erasing in particular the corresponding voltage angle dynamics.
In Section~\ref{sec:Numerics}  we therefore comment on the validity
of our approach and mention numerical results obtained on the full, non-reduced network which 
corroborate our Kron-reduction approach.

Let $\mathcal{N}_g=\{1,\ldots,g\}$ and $\mathcal{N}_c=\{g+1,\ldots,N\}$ be the node subsets representing synchronous machines
and passive nodes respectively, we rewrite the DC power flow equation in the physical network as
\begin{equation}\label{eq:DC PF}
  \bigg[\begin{IEEEeqnarraybox*}[][c]{,c,}
        {\bm P_g} \\
        {\bm P_c} 
      \end{IEEEeqnarraybox*}\bigg]= 
      \bm L_{\bm b}^\textrm{ph}
      \bigg[\begin{IEEEeqnarraybox*}[][c]{,c,}
      \bm \theta_g^\star  \\
      \bm \theta_c^\star 
      \end{IEEEeqnarraybox*}\bigg]=      
      \bigg[\begin{IEEEeqnarraybox*}[][c]{,c/c,}
        \bm L^{gg}_{\bm b} & \bm L^{gc}_{\bm b} \\
        \bm L^{cg}_{\bm b} & \bm L^{cc}_{\bm b}
      \end{IEEEeqnarraybox*}\bigg]
      \bigg[\begin{IEEEeqnarraybox*}[][c]{,c,}
      \bm \theta_g^\star  \\
      \bm \theta_c^\star 
      \end{IEEEeqnarraybox*}\bigg]
      \, .
\end{equation}
Applying Kron reduction to \eqref{eq:DC PF} to eliminate $\bm \theta_c^\star$ yields
\begin{equation}\label{eq:kron reduced DC PF}
  \underbrace{{\bm P}_g - {\bm L}_{\bm b}^{gc}{{\bm L}_{\bm b}^{cc}}^{-1}{\bm P}_c}_{\bm P_\textrm{red}} = 
 \underbrace{\left[ {{\bm L}_{\bm b}^{gg}} - {\bm L}_{\bm b}^{gc}{{\bm L}_{\bm b}^{cc}}^{-1}{{\bm L}_{\bm b}^{cg}} \right]}_{\bm L_\textrm{red}}{\bm \theta}_g^\star \, .
\end{equation}
 Equation \eqref{eq:kron reduced DC PF}
 defines an effective vector of power injections ${\bm P_\textrm{red}}$ and an effective $g \times g$ 
 Laplacian ${\bm L_\textrm{red}}$ on a reduced graph with $g$ nodes \cite{Dorfler13Kron}.
The swing dynamics on the reduced graph reads
\begin{equation}\label{eq:Swing kron}
{\bm M}\ddot{{\bm \theta}}_g = -{\bm D}\dot{{\bm \theta}}_g + {\bm P}_\textrm{red} - {{\bm L}}_\textrm{red}{\bm \theta}_g,
\end{equation}
with ${\bm M}=\textrm{diag}(\{m_i\})$ and $\bm D = \textrm{diag}(\{d_i\})$ for $i\in\mathcal{N}_g$. After Kron reduction, all remaining nodes have inertia. 

Next we show how the observability Gramian formalism outlined earlier can be applied to more general perturbations than the power injection
fluctuations discussed so far. We consider nonsingular single line faults that do not 
split the physical network into  disconnected parts, and introduce a time dependent network Laplacian
\begin{equation}\label{eq: Laplacian N-1}
 {{\bm L}^\textrm{ph}_{\bm b}}(t)={{\bm L}^\textrm{ph}_{\bm b}}-\delta(t) \, b_{\alpha\beta} \, \tau \, {\bm e}^{\phantom{\top}}_{(\alpha,\beta)}{{\bm e}^\top_{(\alpha,\beta)}}\,,
\end{equation}
where ${\bm e}_{(\alpha,\beta)}\in\mathbb{R}^{|\mathcal{N}_g|+|\mathcal{N}_c|}$ and $\tau$ is introduced for 
dimensionality purposes. This describes an infinitesimally short fault at $t=0$, 
during which the $\alpha-\beta$ line breaks. Accordingly, the time-dependent Laplacian \eqref{eq: Laplacian N-1}
 corresponds to a network without the $\alpha-\beta$ edge at $t=0$.
In what follows we consider power networks and operating conditions such that line contingencies do not
drive the transient dynamics outside the basin of attraction of the nominal operating point.
This assumption is realistic when treating short line faults as we do here.
Furthermore, we assume that the linearized swing dynamics still holds under such contingency.
Our numerical investigations to be presented below show that these assumptions are 
justified as long as the fault duration is not too large.

To characterize the transient resulting from a line contingency, we first need to formulate 
how a line fault in the physical network \eqref{eq: Laplacian N-1}, impacts the swing dynamics 
of the Kron reduced network \eqref{eq:Swing kron}.
This requires to distinguish three cases:

\subsubsection{The faulted line connects two synchronous machines}
In this case $\alpha,\beta\in\mathcal{N}_g$ and the fault \eqref{eq: Laplacian N-1} only affects the ${\bm L}_{\bm b}^{gg}$ block 
of the Laplacian ${\bm L}^\textrm{ph}_{\bm b}$.
In terms of ${\bm L}_\textrm{red}$ and ${\bm P}_\textrm{red}$ the fault is described by
\begin{equation}\label{eq:Pred Lred case 1}
  {\bm P}_\textrm{red}\rightarrow \bm P_\textrm{red}\,, \quad \quad
  \bm L_\textrm{red} \rightarrow {\bm L}_\textrm{red} - \delta(t)b_{\alpha\beta} \, \tau \, {\bm e}^{\phantom{\top}}_{(\alpha,\beta)}{{\bm e}^\top_{(\alpha,\beta)}}\,,
\end{equation}
where ${\bm e}_{(\alpha,\beta)}\in\mathbb{R}^{|\mathcal{N}_g|}$.
The swing equation~\eqref{eq:Swing kron}, relative to the nominal operating point $\bm\theta_g(t)=\bm\theta_g^\star+{\bm \varphi}_g(t)$, becomes
\begin{equation}\label{eq:Swing N-1}
 {\bm M}{\ddot{\bm \varphi}}_g = -{\bm D}{\dot{\bm\varphi}}_g - {{\bm L}_\textrm{red}}{\bm \varphi}_g +
 \delta(t) b_{\alpha\beta}\, \tau \, {\bm e}^{\phantom{\top}}_{(\alpha,\beta)}{{\bm e}^\top_{(\alpha,\beta)}}({\bm \theta}_g^\star + \bm \varphi_g)\,.\qquad
 \end{equation}
With the initial condition $({\bm \varphi}_g(0),\bm \omega_g(0))=(0,0)$, the solution to \eqref{eq:Swing N-1} is 
 \begin{equation}\label{eq:Sol N-1 case1}
\begin{IEEEeqnarraybox*}[][c]{c,c,c}   
  \bigg[\begin{IEEEeqnarraybox*}[][c]{,c,}
         \overline{\bm \varphi}_g(t) \\
         \overline{\bm \omega}_g(t) 
       \end{IEEEeqnarraybox*}\bigg]&=&
       e^{{\bm A} t}
       \underbrace{\bigg[\begin{IEEEeqnarraybox*}[][c]{,c,}
         0 \\
         {\bm M}^{-1/2}b_{\alpha\beta} \, \tau \, {\bm e}^{\phantom{\top}}_{(\alpha,\beta)}{{\bm e}^\top_{(\alpha,\beta)}}{\bm \theta}_g^\star
       \end{IEEEeqnarraybox*}\bigg]}_{\bm B}\,,
      \end{IEEEeqnarraybox*}
 \end{equation}
 where $\overline{{\bm \varphi}}_g={\bm M}^{1/2}{\bm \varphi}_g$, 
 $\overline{{\bm \omega}}_g={\bm M}^{1/2}{\bm \omega}_g$, and $\bm A$ is the matrix defined in \eqref{eq:Swing3} with 
 ${\bm L}_{\bm b}$ replaced by ${\bm L}_\textrm{red}$.

\subsubsection{The faulted line connects two passive nodes}
In this case $\alpha, \beta \in\mathcal{N}_c$ and the fault only affects the ${\bm L}_{\bm b}^{cc}$ block of the Laplacian ${\bm L}^\textrm{ph}_{\bm b}$
while the blocks ${\bm L}_{\bm b}^{gg}$, ${\bm L}_{\bm b}^{gc}$ and ${\bm L}_{\bm b}^{cg}$ remain unchanged.
This impacts both ${\bm L}_\textrm{red}$ and ${\bm P}_\textrm{red}$ which become
\begin{IEEEeqnarray}{l}\label{eq:Pred Lred case 2}
  \nonumber
  {\bm P}_\textrm{red} - \delta(t)b_{\alpha\beta} \, \tau \, 
  \frac{{\bm L}_{\bm b}^{gc}[{\bm L}_{\bm b}^{cc}]^{-1}{\bm e}^{\phantom{\top}}_{(\alpha,\beta)}{{\bm e}^\top_{(\alpha,\beta)}}[{\bm L}_{\bm b}^{cc}]^{-1}{\bm P}_c}{1-b_{\alpha\beta}{{\bm e}^\top_{(\alpha,\beta)}} [{\bm L}_{\bm b}^{cc}]^{-1}{\bm e}^{\phantom{\top}}_{(\alpha,\beta)}}\,,\\
  \nonumber
  {\bm L}_\textrm{red} - \delta(t)b_{\alpha\beta} \, \tau \, 
  \frac{{\bm L}_{\bm b}^{gc}[{\bm L}_{\bm b}^{cc}]^{-1}{\bm e}^{\phantom{\top}}_{(\alpha,\beta)}{{\bm e}^\top_{(\alpha,\beta)}}[{\bm L}_{\bm b}^{cc}]^{-1}{\bm L}_{\bm b}^{cg}}{1-b_{\alpha\beta}{{\bm e}^\top_{(\alpha,\beta)}} [{\bm L}_{\bm b}^{cc}]^{-1}{\bm e}^{\phantom{\top}}_{(\alpha,\beta)}}\,,\\  
\end{IEEEeqnarray}
where ${\bm e}_{(\alpha,\beta)}\in\mathbb{R}^{|\mathcal{N}_c|}$, and where we used the Sherman-Morrison formula \cite{Sherman50}
\begin{IEEEeqnarray}{ll}
\nonumber
   [{\bm L}_{\bm b}^{cc}&-b_{\alpha\beta}{\bm e}^{\phantom{\top}}_{(\alpha,\beta)}{{\bm e}^\top_{(\alpha,\beta)}}]^{-1}=
 { [{\bm L}_{\bm b}^{cc}]^{-1}} \\
 & +b_{\alpha\beta}\frac{[{\bm L}_{\bm b}^{cc}]^{-1}{\bm e}^{\phantom{\top}}_{(\alpha,\beta)}{{\bm e}^\top_{(\alpha,\beta)}}[{\bm L}_{\bm b}^{cc}]^{-1}}{1-b_{\alpha\beta}{{\bm e}^\top_{(\alpha,\beta)}} [{\bm L}_{\bm b}^{cc}]^{-1}{\bm e}^{\phantom{\top}}_{(\alpha,\beta)}}\,,
\end{IEEEeqnarray}
to express the inverse of the rank-1 perturbation of ${{\bm L}_{\bm b}^{cc}}$.
Injecting \eqref{eq:Pred Lred case 2} in \eqref{eq:Swing kron}, and solving the swing equation 
with initial conditions $({\bm \varphi}_g(0),\bm \omega_g(0))=(0,0)$ yields
\begin{equation}\label{eq:Sol N-1 case2}
\begin{IEEEeqnarraybox*}[][c]{c,c,c}
  \bigg[\begin{IEEEeqnarraybox*}[][c]{,c,}
         \overline{\bm \varphi}_g(t) \\
         \overline{\bm \omega}_g(t) 
       \end{IEEEeqnarraybox*}\bigg] &=&
  e^{{\bm A}t} \underbrace{\left[
  \begin{IEEEeqnarraybox*}[\footnotesize][c]{,c,}
   0 \\[2mm]
   \frac{-b_{\alpha\beta} \, \tau \, {\bm M}^{-1/2}{{\bm L}_{\bm b}^{gc}}[{\bm L}_{\bm b}^{cc}]^{-1}{\bm e}^{\phantom{\top}}_{(\alpha,\beta)}{{\bm e}^\top_{(\alpha,\beta)}} {\bm \theta}_c^\star}{1-b_{\alpha\beta}{{\bm e}^\top_{(\alpha,\beta)}} [{\bm L}_{\bm b}^{cc}]^{-1}{\bm e}^{\phantom{\top}}_{(\alpha,\beta)}}\\
  \end{IEEEeqnarraybox*}
 \right]}_{\bm B}\,,
\end{IEEEeqnarraybox*}
\end{equation}
where $\overline{{\bm \varphi}}_g={\bm M}^{1/2}{\bm \varphi}_g$, 
$\overline{{\bm \omega}}_g={\bm M}^{1/2}{\bm \omega}_g$, and $\bm A$ is the matrix defined in \eqref{eq:Swing3} with ${\bm L}_{\bm b}$
replaced by ${\bm L}_\textrm{red}$.

\subsubsection{The faulted line connects a synchronous machine and a passive node}
In this case $\alpha\in\mathcal{N}_g$ and $\beta\in\mathcal{N}_c$, and the four blocks of ${\bm L}^\textrm{ph}_{\bm b}$
change according to
\begin{IEEEeqnarray}{lll}
\nonumber
  {\bm L}_{\bm b}^{gg}&\rightarrow&{\bm L}_{\bm b}^{gg}-b_{\alpha\beta}{\bm e}_{\alpha}{\bm e}_{\alpha}^\top \,,\\ \nonumber
  {\bm L}_{\bm b}^{cc}&\rightarrow&{\bm L}_{\bm b}^{cc}-b_{\alpha\beta}{\bm e}_{\beta} {\bm e}_{\beta}^\top  \,,\\ \nonumber
  {\bm L}_{\bm b}^{cg}&\rightarrow&{\bm L}_{\bm b}^{cg}+b_{\alpha\beta}{\bm e}_{\beta} {\bm e}_{\alpha}^\top \,,\\
  {\bm L}_{\bm b}^{gc}&\rightarrow&{\bm L}_{\bm b}^{gc}+b_{\alpha\beta}{\bm e}_{\alpha}{\bm e}_{\beta}^\top \,,
\end{IEEEeqnarray}
where ${\bm e}_{\alpha}\in\mathbb{R}^{|\mathcal{N}_g|}$ and ${\bm e}_{\beta}\in\mathbb{R}^{|\mathcal{N}_c|}$.
${\bm P}_\textrm{red}$ and ${\bm L}_\textrm{red}$ become
\begin{IEEEeqnarray}[\small]{l}\label{eq:Pred Lred case 3}
\nonumber
{\bm P}_\textrm{red} - \delta(t)b_{\alpha\beta} \, \tau \, \frac{[{\bm e}_{\alpha}+{\bm L}_{\bm b}^{gc}[{\bm L}_{\bm b}^{cc}]^{-1}{\bm e}_{\beta}]{\bm e}_{\beta}^\top[{\bm L}_{\bm b}^{cc}]^{-1} {\bm P}_c}{1-b_{\alpha\beta}[L_{b}^{cc}]^{-1}_{\beta\beta}}\,, \\[1mm]
\nonumber
{\bm L}_\textrm{red} - \delta(t)b_{\alpha\beta} \, \tau \, \frac{[{\bm e}_{\alpha}+{\bm L}_{\bm b}^{gc}[{\bm L}_{\bm b}^{cc}]^{-1}{\bm e}_{\beta}][{\bm e}_{\alpha}^\top+{\bm e}_{\beta}^\top[{\bm L}_{\bm b}^{cc}]^{-1}{\bm L}_{\bm b}^{cg}]}{1-b_{\alpha\beta}[L_{\bm b}^{cc}]^{-1}_{\beta\beta}}\,, \\
\end{IEEEeqnarray}
where, again, we used the Sherman-Morrison formula \cite{Sherman50} to compute the inverse of 
${\bm L}_{\bm b}^{cc} -b_{\alpha\beta}{\bm e}_{\beta} {\bm e}_{\beta}^\top$.
Finally, injecting \eqref{eq:Pred Lred case 3} into \eqref{eq:Swing kron}, and solving the swing equation 
with initial conditions $({\bm \varphi}_g(0),\bm \omega_g(0))=(0,0)$ yields
\begin{equation}\label{eq:Sol N-1 case3}
\begin{IEEEeqnarraybox*}[][c]{c,c,c}
\bigg[\begin{IEEEeqnarraybox*}[][c]{,c,}
         \overline{\bm \varphi}_g(t) \\
         \overline{\bm \omega}_g(t) 
       \end{IEEEeqnarraybox*}\bigg] &=&
  e^{{\bm A}t} \underbrace{
  \left[\begin{IEEEeqnarraybox*}[\scriptsize][c]{,c,}
   0 \\[1mm]
   \frac{b_{\alpha\beta} \, \tau \, {\bm M}^{-1/2}({\bm e}_{\alpha}+{\bm L}_{\bm b}^{gc}[{\bm L}_{\bm b}^{cc}]^{-1}{\bm e}_{\beta})(\theta_{g,\alpha}^\star-\theta_{c,\beta}^\star)}{1-b_{\alpha\beta}[{\bm L}_{\bm b}^{cc}]^{-1}_{\beta\beta}}\\
   \end{IEEEeqnarraybox*}\right]}_{\bm B}\,, 
\end{IEEEeqnarraybox*}
\end{equation}
 where $\overline{{\bm \varphi}}_g={\bm M}^{1/2}{\bm \varphi}_g$, 
 $\overline{{\bm \omega}}_g={\bm M}^{1/2}{\bm \omega}_g$, and $\bm A$ is the matrix defined in \eqref{eq:Swing3} with ${\bm L}_{\bm b}$
 replaced by ${\bm L}_\textrm{red}$.

Having solved the swing equation for the three types of line faults we are now ready to calculate performance measures
and present our main results.

\begin{proposition}[\bf{angle coherence under line contingency}]\label{prop:angle coherence contingency}
Consider the Kron reduced power system model of \eqref{eq:Swing kron} with
$\bm L_{\rm red}\rightarrow \mathbb L_{\rm red}(\epsilon) = \bm L_{\rm red} +\epsilon \mathbb I$, $\epsilon>0$ and
$m_i=m=d /\gamma ~\forall i\in\mathcal{N}_g$.
The angle coherence measure 
\begin{equation}
\mathcal{P}_{\bm \varphi} = \int_0^\infty {\bm \varphi}_g^\top{\bm \varphi}_g \dd t \, ,
\end{equation} 
evaluated for a line contingency given by \eqref{eq: Laplacian N-1} is 
\begin{equation}\label{eq:L2 coherence gen-gen contingency}
\mathcal{P}_{ \bm \varphi} = \frac{{P^2_{\alpha,\beta}}\, \tau^2}{2d}\Omega_{\alpha\beta}\,,
\end{equation}
if the faulted line connects two synchronous machines, by
\begin{equation}\label{eq:L2 coherence con-con contingency}
\mathcal{P}_{\bm \varphi}=\frac{{P^2_{\alpha,\beta}}\, \tau^2}{2d}
 \frac{\Omega_{\alpha\beta}-{\bm e}^\top_{(\alpha,\beta)}[{\bm L}_{\bm b}^{cc}]^{-1}{\bm e}^{\phantom{\top}}_{(\alpha,\beta)}}{[1-b_{\alpha\beta}{\bm e}^\top_{(\alpha,\beta)}[{\bm L}_{\bm b}^{cc}]^{-1}{\bm e}^{\phantom{\top}}_{(\alpha,\beta)}]^2}\,,
\end{equation}
if the faulted line connects two passive nodes, and by
\begin{equation}\label{eq:L2 coherence gen-con contingency}
 \mathcal{P}_{ \bm \varphi}=\frac{{P^2_{\alpha,\beta}}\, \tau^2}{2d}
 \frac{\Omega_{\alpha\beta}-[{\bm L}_{\bm b}^{cc}]^{-1}_{\beta\beta}}{[1-b_{\alpha\beta}[{\bm L}_{\bm b}^{cc}]^{-1}_{\beta\beta}]^2}\,,
\end{equation} 
if the faulted line connects a synchronous machine $\alpha$ and a passive node $\beta$.
In \eqref{eq:L2 coherence gen-gen contingency}--\eqref{eq:L2 coherence gen-con contingency}, 
$P_{\alpha,\beta}$ is the power flow on the $\alpha-\beta$ 
line prior to the fault, and $\Omega_{\alpha\beta}$ is the resistance distance computed with respect 
to the physical network ${\bm L}^\textrm{ph}_{\bm b}$, prior to Kron reduction.
\end{proposition}
\begin{proof}
The Kron reduced Laplacian $\bm L_{\rm red}$ of \eqref{eq:kron reduced DC PF} has a marginally stable mode which 
we deal with using the regularization scheme described above, i.e. by introducing 
$\mathbb L_{\rm red}(\epsilon) = \bm L_{\rm red} +\epsilon \mathbb I$, with the $g \times g$ identity matrix and $\epsilon >0$.
The observability Gramian associated with the angle coherence measure is obtained 
from \eqref{eq:X22 angle} as $\mathbb L_{\rm red}^{-1}(\epsilon)/2 \gamma$ [see \eqref{eq:Gramian coherence}]. 
Thanks to the regularization, this is well defined for $\epsilon >0$ (see
Propositions \ref{prop:Regularized inverse of L} and \ref{prop:Marginal mode assumption}). 
To compute $\mathcal{P}={\bm B}^\top \bm X {\bm B}$ for the three types of line contingencies, we use $\bm B$ defined in \eqref{eq:Sol N-1 case1},
\eqref{eq:Sol N-1 case2} and \eqref{eq:Sol N-1 case3} respectively.
We will see at the end of the calculation that, for the line fault \eqref{eq: Laplacian N-1}, 
the limit $\epsilon \rightarrow 0$ is also well-behaved.

For a line fault between two synchronous machines, one obtains rather directly
\begin{IEEEeqnarray}{c}\label{eq:calP}
\mathcal{P}_{\bm \varphi}=\frac{b_{\alpha \beta}^2}{2 d} (\theta_{g,\alpha}^\star-\theta_{g,\beta}^\star)^2 \, \tau^2 \,
\left[ \bm e_{(\alpha,\beta)}^\top \mathbb L_{\rm red}^{-1}(\epsilon) \bm e_{(\alpha,\beta)}
\right] \, ,
\end{IEEEeqnarray}
where the symbols $e_{(\alpha,\beta)}$ have been defined in Section \ref{sec:Notations}.
To calculate the expression in the square bracket we introduce the matrix $\bm T_{\rm red}$ diagonalizing 
$\mathbb L_{\rm red}(\epsilon)$, 
i.e. $\bm T^\top_{\rm red} \mathbb L_{\rm red}(\epsilon) \bm T_{\rm red} = {\rm diag}(\{\lambda_i+\epsilon\})$.
Because $\mathbb L_{\rm red}(\epsilon)$ differs from $\bm L_{\rm red}$
only by a multiple of the identity, $\bm T_{\rm red}$ also diagonalizes $\bm L_{\rm red}$ and is therefore independent of $\epsilon$.
One obtains 
\begin{IEEEeqnarray}{ccc}\label{eq:sqbracket}
\bm e_{(\alpha,\beta)}^\top \mathbb L_{\rm red}^{-1}(\epsilon) \bm e_{(\alpha,\beta)} &=& 
\sum_{l=1}^g \frac{[(T_{\rm red})_{\alpha l} - (T_{\rm red})_{\beta l} ]^2}{\lambda_l+\epsilon} \nonumber \\
&=& \sum_{l=2}^g \frac{[(T_{\rm red})_{\alpha l} - (T_{\rm red})_{\beta l} ]^2}{\lambda_l+\epsilon} \, ,
\end{IEEEeqnarray}
where in the last line we omitted the $l=1$ contribution, because the $(T_{\rm red})_{\alpha 1} - (T_{\rm red})_{\beta 1}=0$
since the marginally stable mode has constant components. We can then take the limit $\epsilon \rightarrow 0$ without
problem. Since $ \bm T_{\rm red}$ is independent of $\epsilon$ (see above), 
inserting \eqref{eq:sqbracket} into \eqref{eq:calP} and noting that $P_{\alpha,\beta} = 
b_{\alpha \beta} (\theta_{g,\alpha}^\star-\theta_{g,\beta}^\star)$ 
finally gives \eqref{eq:L2 coherence gen-gen contingency}.

The calculation is more intricate for a line fault between two passive nodes in the physical network. 
Using \eqref{eq:Pred Lred case 2} and \eqref{eq:Sol N-1 case2} one obtains, after some algebra
\begin{IEEEeqnarray}{ccc}\label{eq:calPnn}
\mathcal{P}_{\bm \varphi}&=&\frac{b_{\alpha \beta}^2 (\theta_{c,\alpha}^\star-\theta_{c,\beta}^\star)^2  \, \tau^2 }{2 d} \,  \frac{1}{(1-
b_{\alpha \beta} \bm e_{(\alpha,\beta)}^\top [\bm L_{\bm b}^{\rm cc}]^{-1} \bm e_{(\alpha,\beta)})^2} \nonumber  \\
&\times &
\left[ \bm e_{(\alpha,\beta)}^\top (\bm L_{\bm b}^{\rm cc})^{-1} \bm L_{\bm b}^{\rm cg} \mathbb L_{\rm red}^{-1}(\epsilon) 
\bm L_{\bm b}^{\rm gc} [\bm L_{\bm b}^{\rm cc}]^{-1} \bm e_{(\alpha,\beta)}
\right] \, , \quad
\end{IEEEeqnarray}
which depends on blocks of the physical Laplacian as well as inverses thereof and of $\mathbb L_{\rm red}(\epsilon)$. The regularization  
$\bm L_{\rm red} \rightarrow  \mathbb L_{\rm red}(\epsilon) = \bm L_{\rm red} +\epsilon \mathbb I$ is defined on the Kron reduced Laplacian.
For the physical Laplacian, it reads
\begin{equation}\label{eq:DC PF2}
  \bm L^\textrm{ph}_{\bm b} \rightarrow \mathbb L^\textrm{ph}_{\bm b}(\epsilon) =  
  \bigg[\begin{IEEEeqnarraybox*}[][c]{,c/c,}
        \bm L^{gg}_{\bm b} & \bm L^{gc}_{\bm b} \\
        \bm L^{cg}_{\bm b} & \bm L^{cc}_{\bm b}
      \end{IEEEeqnarraybox*}\bigg]
      +  \bigg[\begin{IEEEeqnarraybox*}[][c]{,c/c,}
       \epsilon \mathbb I &0 \\
        0& 0
      \end{IEEEeqnarraybox*}\bigg]
      \, .
\end{equation}
Using matrix block inversion \cite{Horn2012matrix} and \eqref{eq:kron reduced DC PF}
one obtains
\begin{IEEEeqnarray}{c}\label{eq:Block inversion property}
\nonumber
 [(\mathbb L^\textrm{ph}_{\bm b}(\epsilon))^{-1}]^{cc} = [\bm L_{\bm b}^{cc}]^{-1} + 
 [\bm L_{\bm b}^{cc}]^{-1}\bm L_{\bm b}^{cg} ({\mathbb L}_\textrm{red}(\epsilon))^{-1} \bm L_{\bm b}^{gc}[\bm L_{\bm b}^{cc}]^{-1}\, 
\end{IEEEeqnarray}
with which we rewrite \eqref{eq:calPnn} as
\begin{IEEEeqnarray}{ccc}\label{eq:PPP}
\mathcal{P}_{\bm \varphi}&=&\frac{b_{\alpha \beta}^2 (\theta_{c,\alpha}^\star-\theta_{c,\beta}^\star)^2  \, \tau^2 }{2 d} \,  
\frac{\bm e_{(\alpha,\beta)}^\top ( [(\mathbb L^\textrm{ph}_{\bm b}(\epsilon))^{-1}]^{cc} - [\bm L_{\bm b}^{cc}]^{-1}) \bm e_{(\alpha,\beta)}}{1-
b_{\alpha \beta} \bm e_{(\alpha,\beta)}^\top [\bm L_{\bm b}^{\rm cc}]^{-1}(\epsilon) \bm e_{(\alpha,\beta)}}  \, . \qquad
\end{IEEEeqnarray}
We still need to calculate $\lim_{\epsilon \rightarrow 0} \, \bm e_{(\alpha,\beta)}^\top  [(\mathbb L^{\textrm ph}_{\bm b}(\epsilon) )^{-1}]^{cc}
\bm e_{(\alpha,\beta)}$. We do this using matrix perturbation theory~\cite{Stewart},
which  expands   eigenvalues and eigenvectors of $\mathbb L^\textrm{ph}_{\bm b}(\epsilon)$
in a power series in the small parameter $\epsilon$. From \eqref{eq:DC PF2}
we obtain, to leading order in $\epsilon$~\cite{Stewart},
\begin{IEEEeqnarray}{ccc}\label{eq:1}
\lambda_l(\epsilon) &=& \lambda_l(0) + \epsilon \sum_{i=1}^g [u_i^{(l)}(0)]^2 + {\cal O}(\epsilon^2) \, , \\ \label{eq:2}
u_i^{(l)}(\epsilon) &=& u_i^{(l)}(0) + \epsilon \sum_{l' \ne l} C_{l,l'} \, u_i^{(l')}(0) + {\cal O}(\epsilon^2)\, ,  \\ \label{eq:3}
C_{l,l'} &=& \frac{\sum_{j=1}^g u_j^{(l')}(0) \, u_j^{(l)}(0)}{\lambda_l(0)-\lambda_{l'}(0)} 
\end{IEEEeqnarray}
in terms of the eigenvectors $\bm u^{(l)}(0)$ of $\bm L_{\bm b}^\textrm{ph}$.
We use this to write 
\begin{IEEEeqnarray}{ccc}\label{eq:invL}
\mathbb L^\textrm{ph}_{\bm b}(\epsilon)^{-1}& = & \sum_{l \ge 1} \frac{\bm u^{(l)}(\epsilon) \, (\bm u^{(l)}(\epsilon))^\top}{\lambda_l(\epsilon)} \, .
\end{IEEEeqnarray}
Equation \eqref{eq:invL} finally allows us to calculate the contribution of the $l=1$ marginally stable mode to 
$\bm e_{(\alpha,\beta)}^\top  [(\mathbb L^\textrm{ph}_{\bm b}(\epsilon))^{-1}]^{cc} \bm e_{(\alpha,\beta)}$. Using \eqref{eq:1}--\eqref{eq:3}
and $\bm e_{(\alpha,\beta)}^\top \, \bm u^{(1)}(0)=0$, that contribution is 
\begin{IEEEeqnarray}{ccc}
\frac{\epsilon^2 \sum_{l',l'' \ne l} C_{l,l'} C_{l,l''}  [e_{(\alpha,\beta)}^\top \, \bm u^{(l')}(0)][(\bm u^{(l'')}(0))^\top \, e_{(\alpha,\beta)}] +{\cal O}(\epsilon^3)}{\epsilon \sum_{i=1}^g [u_i^{(l)}(0)]^2 + {\cal O}(\epsilon^2)} \nonumber ,
\end{IEEEeqnarray}
which goes linearly to zero with $\epsilon \rightarrow 0$. Higher order terms not considered in \eqref{eq:1}--\eqref{eq:3}
vanish even faster and accordingly they have been neglected. 
To calculate $\bm e_{(\alpha,\beta)}^\top  [(\mathbb L^\textrm{ph}_{\bm b})^{-1}]^{cc} \bm e_{(\alpha,\beta)}$ in \eqref{eq:PPP}
the $l=1$ term can therefore be omitted in \eqref{eq:invL}, and one finally obtains
\begin{IEEEeqnarray}{c}
\lim_{\epsilon \rightarrow 0} \bm e_{(\alpha,\beta)}^\top  [(\mathbb L^\textrm{ph}_{\bm b}(\epsilon))^{-1}]^{cc} \bm e_{(\alpha,\beta)} =
\Omega_{\alpha \beta} \, .
\end{IEEEeqnarray}
Inserting this into \eqref{eq:PPP} and noting that $P_{\alpha,\beta} = 
b_{\alpha \beta} (\theta_{c,\alpha}^\star-\theta_{c,\beta}^\star)$ one finally obtains \eqref{eq:L2 coherence gen-gen contingency}. 

The proof of the Proposition for a line fault between a passive and an active node proceeds along similar steps as
for a line fault between two passive nodes in the physical network, and we do
not present it here. 
\end{proof}

The results of Proposition \ref{prop:angle coherence contingency} show that for the three types of line contingencies, the voltage angle deviation from the nominal operating 
point is proportional to the square of the power flowing on the line prior to the fault, times a topological factor. 
The latter is equal to the resistance distance when the faulted line 
connects two synchronous machines \eqref{eq:L2 coherence gen-gen contingency}.
The resistance distance $\Omega_{\alpha\beta}$, and accordingly the response $\mathcal{P}$, will be greater for lines
such that there are few alternative paths connecting $\alpha$ to $\beta$ beyond the direct line.
For the other two types of line faults, \eqref{eq:L2 coherence con-con contingency} and \eqref{eq:L2 coherence gen-con contingency},
the resistance distance factor of \eqref{eq:L2 coherence gen-gen contingency}
is complemented by terms which account for the topology of the network of passive nodes, with no straightforward 
interpretation, except for the term
${\bm e}^\top_{(\alpha,\beta)}[{\bm L}_{\bm b}^{cc}]^{-1}{\bm e}^{\phantom{\top}}_{(\alpha,\beta)}$
in \eqref{eq:L2 coherence con-con contingency}, which is equal to the resistance distance between $\alpha$ and $\beta$
in the network of passive nodes augmented by a ground node (see Thm.~3.9 of Ref.~\cite{Dorfler13Kron}).
Finally, the denominator in \eqref{eq:L2 coherence con-con contingency} is much smaller than one, yielding large responses 
$\mathcal{P}$, for lines $\alpha-\beta$ between weakly connected components of the network of passive nodes.
Remember however that the above results are no longer valid for a line fault disconnecting the network.

\begin{figure*}[h]
\centering
  \includegraphics[width=\textwidth]{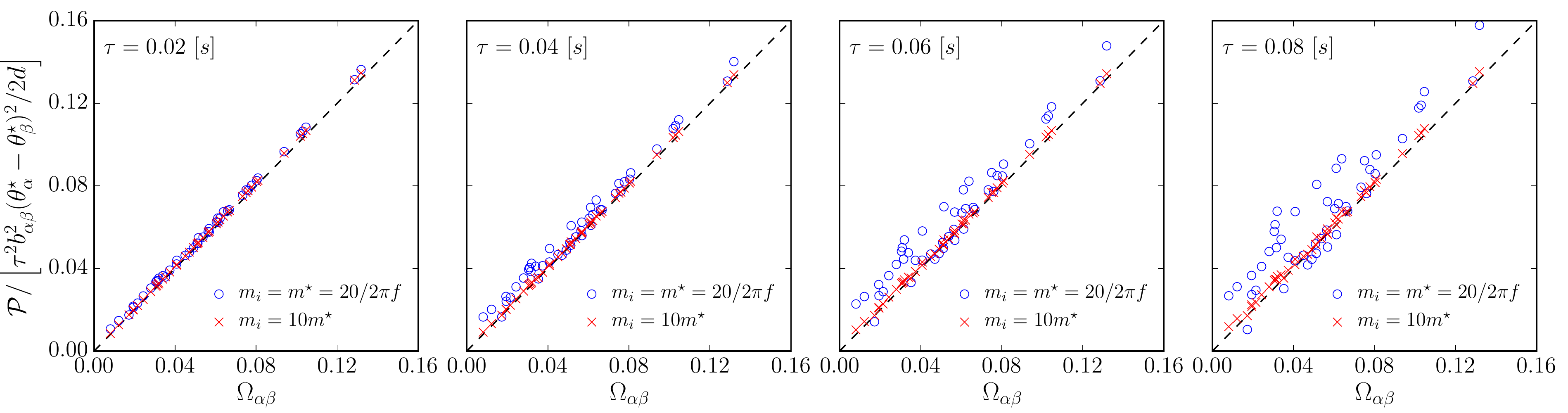}
 \caption{
 Transient angle coherence measure {\color{blue} \eqref{eq:primcontrol} }for a line contingency as a function of the resistance distance separating the 
 nodes of the faulted line.
 Each data point corresponds to the fault of a line connecting two generators in the physical network.
 Simulation parameters: IEEE 118-bus test case with uniform inertia at all nodes, $f=50$ Hz,
 $d_i/m_i=0.5~[s^{-1}]$, and  $m_i=m^\star=2H/2\pi f$, $H=10~[s]$ (typical values from \cite{Kundur}, blue circles)
 and  $m_i = 10m^\star$ (red crosses).
 From left to right: fault clearing times $\tau$ corresponding to $1,2,3,$ and $4$ AC cycles.
 The straight line is the theoretical prediction \eqref{eq:L2 coherence gen-gen contingency}.}
 \label{fig:Figure N-1 angle}
\end{figure*}

\begin{proposition}[\bf{Primary control effort under line contingency}]
Consider the Kron reduced power system model of \eqref{eq:Swing kron} with
$\bm L_{\rm red} \rightarrow L_{\rm red}(\epsilon) = \bm L_{\rm red} +\epsilon \mathbb I$, $\epsilon>0$, and satisfying 
Assumption \ref{ass:Homogeneous inertia damping ratio}.
The primary control effort~\cite{Poola17} 
\begin{equation}
\label{eq:primcontrol}
\mathcal{P}_{\bm  \omega}=\int_0^\infty \sum_id_i\omega_i^2\dd t 
\end{equation}
required during the transient caused by a line contingency modeled 
by \eqref{eq: Laplacian N-1} is given by
\begin{IEEEeqnarray}{ll}\label{eq:L2 primary effort gen-gen contingency}
\mathcal{P}_{\bm  \omega}  = \frac{{P^2_{\alpha,\beta}} \, \tau^2 }{2}(m_\alpha^{-1}+m_\beta^{-1})\,,
\end{IEEEeqnarray}
if the faulted line connects two synchronous machines, by
\begin{IEEEeqnarray}{l}\label{eq:L2 primary effort con-con contingency}
\mathcal{P}_{ \bm  \omega} =\frac{{P^2_{\alpha,\beta}} \, \tau^2 }{2} 
 \frac{\sum_{i\in\mathcal{N}_g}m_i^{-1}{\left[{\bm e}^\top_{(\alpha,\beta)}{[{\bm L}_{\bm b}^{cc}}]^{-1}{\bm L}_{\bm b}^{cg}{\bm e}_i\right]}^2}{[1-b_{\alpha\beta}{\bm e}^\top_{(\alpha,\beta)}[{\bm L}_{\bm b}^{cc}]^{-1}{\bm e}^{\phantom{\top}}_{(\alpha,\beta)}]^2}\,,
 \end{IEEEeqnarray}
if the faulted line connects two passive nodes, and by
\begin{IEEEeqnarray}{l}\label{eq:L2 primary effort gen-con contingency}
\mathcal{P}_{\bm  \omega} =\frac{{P^2_{\alpha,\beta}} \, \tau^2 }{2}
 \frac{\sum_{i\in\mathcal{N}_g}m_i^{-1}{\left[\delta_{i\alpha}+{\bm e}_{\beta}^\top[{{\bm L}_{\bm b}^{cc}}]^{-1}{\bm L}_{\bm b}^{cg}{\bm e}_{i}\right]}^2}
 {[1-b_{\alpha\beta}[{L}_{\bm b}^{cc}]_{\beta\beta}^{-1}]^2}\,,\quad
\end{IEEEeqnarray}
if the faulted line connects a synchronous machine $\alpha$ and a passive node $\beta$.
In \eqref{eq:L2 primary effort gen-gen contingency}--\eqref{eq:L2 primary effort gen-con contingency}, 
$P_{\alpha,\beta}$ is the power flow on the $\alpha-\beta$ line
prior to the fault.
\end{proposition}
\begin{proof}
From \eqref{eq:X22 frequency} with $\bm Q^{(2,2)}= \bm D$, one obtains 
the observability Gramian associated to the primary control effort as $\bm X^{(2,2)}=\mathbb{I}/2$, independently of $\epsilon$.
To compute $\mathcal{P}=\bm B^\top \bm X\bm B$ for the above three types of lines contingencies, we use $\bm B$ defined in \eqref{eq:Sol N-1 case1},
\eqref{eq:Sol N-1 case2} and \eqref{eq:Sol N-1 case3} respectively.
After some straightforward matrix multiplications, and rewriting the diagonal matrix of the synchronous machines inertia
as $\bm M^{-1}=\sum_i m_i^{-1}{{\bm e}}_i{{\bm e}}^\top_i$ for 
${{\bm e}}_i \in \mathbb{R}^{|\mathcal{N}_g|}$, one obtains \eqref{eq:L2 primary effort gen-gen contingency}, \eqref{eq:L2 primary effort con-con contingency}
and \eqref{eq:L2 primary effort gen-con contingency}.
\end{proof}

Equation \eqref{eq:L2 primary effort gen-gen contingency} shows that the effort of primary control which results from the outage of the $\alpha-\beta$ line is proportional to the square of the
power flowing on the line prior to the fault times the prefactor $(m_\alpha^{-1}+m_\beta^{-1})$.
The primary control effort is therefore large if the rotational inertias of the synchronous machines at
both ends of the faulted line are small.
For the other types of line contingencies, \eqref{eq:L2 primary effort con-con contingency} and 
\eqref{eq:L2 primary effort gen-con contingency} predict a more
involved dependence of $\mathcal{P}$ with the inertias of the synchronous machines.
Quite interestingly, for both \eqref{eq:L2 primary effort con-con contingency} and 
\eqref{eq:L2 primary effort gen-con contingency}, only the inertias of the synchronous machines directly connected to the passive 
nodes matter. 
This is easily seen in \eqref{eq:L2 primary effort con-con contingency} and \eqref{eq:L2 primary effort gen-con contingency}
noticing that for ${{\bm e}}_i \in\mathbb{R}^{|\mathcal{N}_g|}$, ${\bm L}_{\bm b}^{cg}{{\bm e}}_i=0$
if the $i^\textrm{th}$ synchronous machine is not connected to any of the passive nodes.
Furthermore, the contribution of the inertias of the synchronous machines connected to the passive nodes is weighted by 
a network topology dependent term $\propto [{{\bm L}_{\bm b}^{cc}}]^{-1}{\bm L}_{\bm b}^{cg}$.

For comparison, we note that the primary control effort to restore synchrony in the case of a power injection 
perturbation only depends on the amplitude of the perturbation and on the inertia of the machine at the perturbed node, 
and not on the steady state operating conditions~\cite{Poola17}.
For line faults, our results show that the performance measure depends on the square of the power flowing on the line
prior to the fault times an inertia dependent contribution containing the inertias connected at the two ends of the faulted line.

\section{Performance measures under line contingencies: Numerical Analysis}\label{sec:Numerics}

\begin{figure*}[!ht]
\centering
  \includegraphics[width=\textwidth]{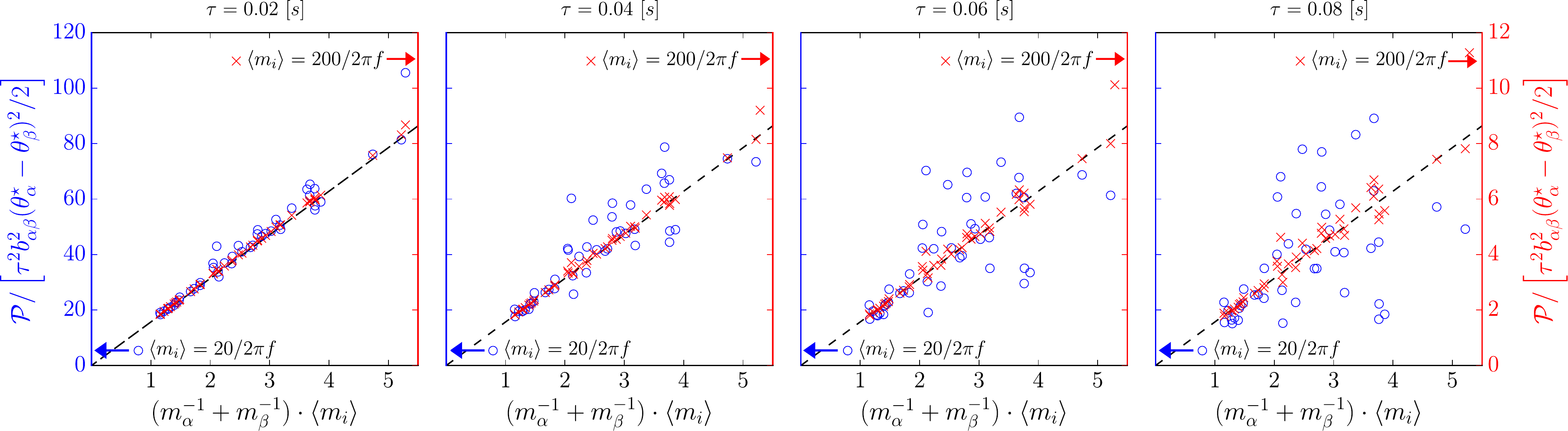}
 \caption{Primary control effort {\color{blue}\eqref{eq:primcontrol} } required during a transient resulting from a 
 line contingency as a function of the sum of 
 the inverse inertias of the synchronous machines at both ends of the faulted line. 
 Each data point corresponds to the fault of a line connecting two generators in the physical network.
 Simulation parameters: IEEE 118-bus test case with inertias uniformly distributed in the interval
 $[0.2 \langle m \rangle,1.8 \langle m \rangle]$ with $\langle m\rangle=2H/2\pi f$, $H=10~[s]$ (typical values from \cite{Kundur}, blue circles, left vertical scale)
 and $\langle m\rangle=200/2\pi f$ (red crosses, right vertical scale), $f=50$ Hz, and $d_i/m_i=0.5~[s^{-1}]$. 
 From left to right: fault clearing times $\tau$ corresponding to $1,2,3,$ and $4$ AC cycles.
 The straight line is the theoretical prediction \eqref{eq:L2 primary effort gen-gen contingency}.}
 \label{fig:Figure N-1 frequency}
\end{figure*}

To illustrate our results we numerically investigate the IEEE 118-bus 
test case \cite{IEEEtestcase}.
We simulate the swing dynamics 
for the reduced model \eqref{eq:Swing kron}, where all PQ buses have been eliminated by Kron reduction.
To model temporary line disconnections, we consider time-dependent network Laplacians,
\begin{equation}\label{eq:Circuit breakers Laplacian}
{\bm L}^\textrm{ph}_{\bm b}(t)={\bm L}^\textrm{ph}_{\bm b}+\Theta(t)\Theta(\tau-t)b_{\alpha\beta}{\bm e}^{\phantom{\top}}_{(\alpha,\beta)}{{\bm e}^\top_{(\alpha,\beta)}}\,,
\end{equation}
where $\Theta(t)$ is the Heavyside step function, $\tau$ is the clearing time, and $\alpha-\beta$
is the faulted line. Because our theoretical predictions were derived for Dirac-$\delta$ perturbations, 
we expect numerical data to confirm our theory for short enough $\tau$. 

We perform numerical simulations for all possible line contingencies in the network.
For each contingency simulation we evaluate numerically the performance measures
$\int_0^\infty\sum_i \varphi_i^2(t) \dd t$ and $\int_0^\infty \sum_i d_i\omega_i^2(t)\dd t$.

Fig.~\ref{fig:Figure N-1 angle} plots the angle coherence measure
rescaled by the square of the power flowing on the line prior to the fault 
for all lines connecting two active nodes in the physical network. 
As predicted by \eqref{eq:L2 coherence gen-gen contingency},
this quantity is linear in the resistance distance.
The validity of the theory for Dirac-$\delta$ perturbations extends to longer clearing times 
for larger inertias (red crosses). 
This is so because for larger inertia, the transient voltage angle oscillations are quickly absorbed relatively close to 
the faulted line and do not have the time to propagate to distant nodes before the fault is cleared.
As a matter of fact, the eigenmode $i$ of the network Laplacian has a characteristic excitation time scale 
$d/\lambda_i=\gamma m/\lambda_i$. When the fault duration is shorter than this time scale
the $i^{\rm th}$ mode is not excited by the fault. Increasing the inertia therefore leads to less and less excited modes, until
$\tau < \gamma m/\lambda_2$ at which point no mode is excited and the disturbance is effectively infinitesimally short.
For our parameters, we find that $\gamma m/\lambda_2=0.025 s$ for the low inertia data and $\gamma m/\lambda_2=0.25 s$ 
for the large inertia data in Fig.~\ref{fig:Figure N-1 angle}. The breakdown of our theory for 
$\tau \gtrsim 0.04 s$ in Fig.~\ref{fig:Figure N-1 angle} therefore indicates that our theory is valid 
for $\tau < \gamma m/\lambda_2$.
These results suggest that for longer clearing times, and more generally for perturbations that are extended in time, 
alternative approaches to the observability Gramian are needed to accurately evaluate performance 
measures~\cite{tyloo2017robustness}. 

Fig.~\ref{fig:Figure N-1 frequency} plots the primary control effort
rescaled by the square of the power flowing on the line prior to the fault for all lines connecting two active nodes
in the physical network. 
For heterogeneous inertias and sufficiently short clearing times 
\eqref{eq:L2 primary effort gen-gen contingency} predicts that this quantity scales linearly with $(m_\alpha^{-1}+m_\beta^{-1})$.
This prediction is confirmed for sufficiently large inertias.
For lower values of inertia, the linear tendency holds for short clearing times, but breaks down 
for longer faults.
Note that the red crosses in Figs.~\ref{fig:Figure N-1 angle} and \ref{fig:Figure N-1 frequency} correspond to 
somehow exaggerated values of inertia. They are here to illustrate that our theoretical predictions  
remain valid for longer fault clearing times.

Finally, Fig.~\ref{fig:detail} shows the angle coherence and the primary control effort measures for all possible line contingencies
(170 lines in total, of which 46 connect two synchronous machines, 35 connect two passive nodes and 89 connect a passive node to a 
synchronous machine) and $\tau=20$~ms. The transmission lines are sorted according to the square of the power flowing
on the line in normal operation. 
Numerical results confirm that our theoretical predictions 
for the angle coherence measure, 
\eqref{eq:L2 coherence gen-gen contingency}--\eqref{eq:L2 coherence gen-con contingency}, as well as 
 for the primary control effort,
\eqref{eq:L2 primary effort gen-gen contingency}--\eqref{eq:L2 primary effort gen-con contingency}. 

Remarkably, the transient performance is not a monotonic function of the 
power load of the faulted line, the square of which is indicated by the black line in Fig.~\ref{fig:detail}. 
We observe that the most critical lines are not always the most highly loaded ones.
For the angle coherence measure (left panel in Fig.~\ref{fig:detail}), the line carrying the $6^\textrm{th}$ largest power leads to the largest
integrated transient excursions even though it carries 30$\%$ less power than the line carrying the most power.
We saw (but do not show here) this non monotonic behavior also when lines are sorted according to their relative load 
(the load relative to their capacity). 
A similar non monotonicity is observed for the primary control effort (right panel in Fig.~\ref{fig:detail}).
The fault of the line carrying the $3^\textrm{rd}$ largest power causes the largest primary 
control effort, and the line carrying 44$\%$ of the largest transmitted power is the $4^\textrm{th}$ most critical one.

\begin{figure*}[!hbt]
\centering
  \includegraphics[width=\columnwidth]{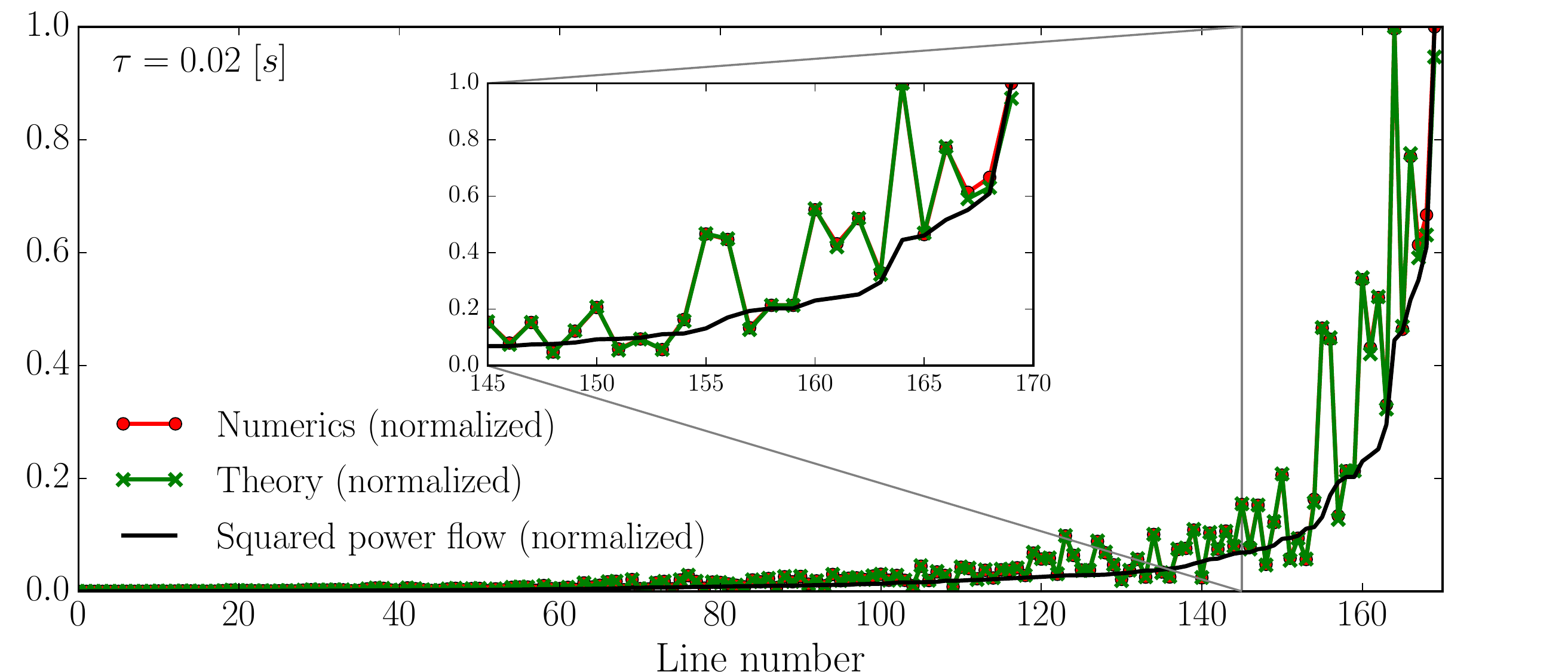}
  \includegraphics[width=\columnwidth]{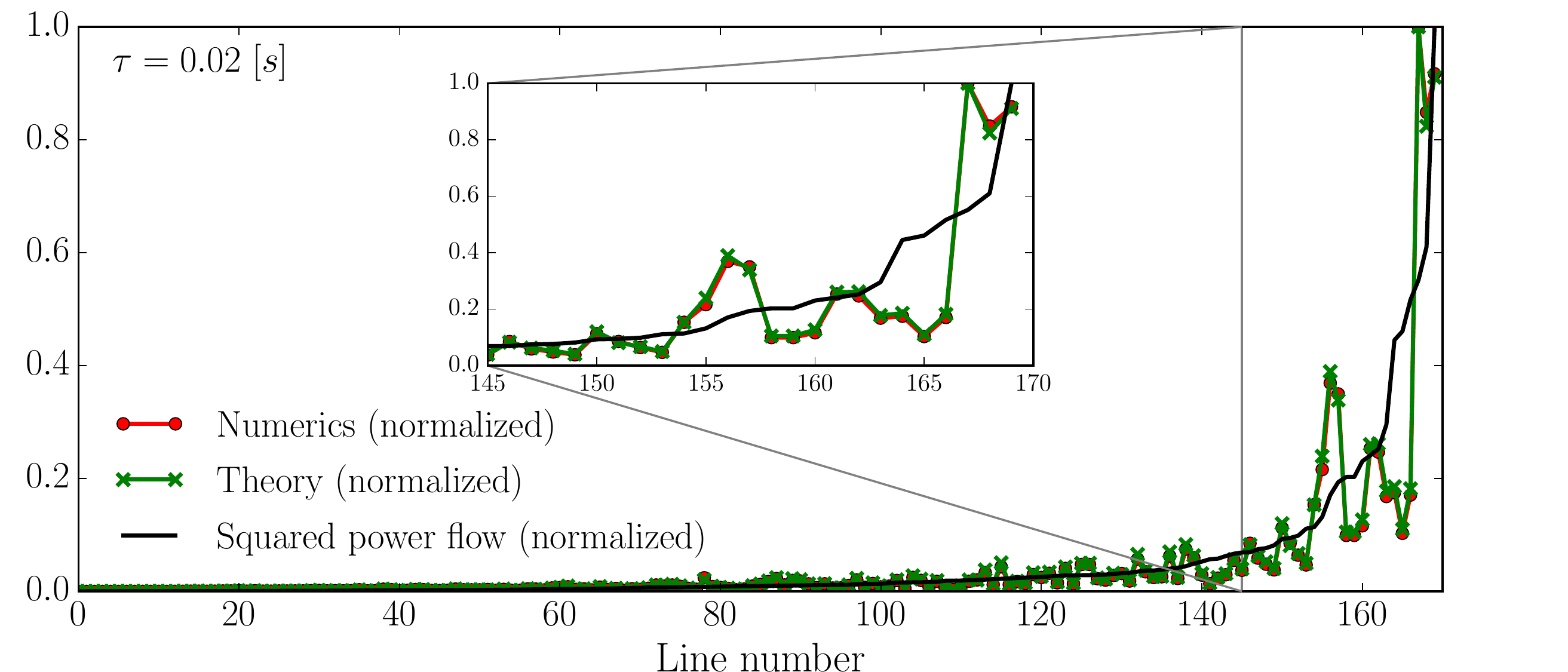}
 \caption{Plot of the normalized performance measure (red), theoretical prediction (green), 
 and square of the power flowing on the line prior to the fault (black) for line contingencies with clearing times $\tau = 0.02~[s]$.
 Transmission lines are ordered according to the power flowing on the line prior to the fault. 
 Left panel: angle coherence measure, uniform inertias $m_i=m^\star=2H/2\pi f$, $H=10~[s]$, $f=50$ Hz, and $d_i/m_i=0.5~[s^{-1}]$.
 Right panel: primary control effort, inertias uniformly distributed in the interval 
 $[0.2 \langle m \rangle,1.8 \langle m \rangle]$ with $\langle m\rangle=2H/2\pi f$, $H=10~[s]$, $f=50$ Hz, and $d_i/m_i=0.5~[s^{-1}]$.}
 \label{fig:detail}
\end{figure*}

We finally note that numerical calculations not shown here 
for line faults of duration $\tau=20 ms$ indicate that performance measures
calculated over the full, non-reduced network corroborate the results presented above for a large range of values
of the damping parameters $d_i$ on the inertialess buses.

\section{Conclusion}\label{sec:Conclusion}

The standard formalism used until now to evaluate performance measures of electric power grids
was restricted to nodal perturbations~\cite{Gayme15,Gayme16,Poola17} and we extended it to line perturbations
using a regularization procedure. The latter  stabilizes the otherwise marginally stable Laplacian mode during the 
calculation and allows us to identify cancellations of divergences -- it suppresses removable singularities in the calculation.
We showed numerically that, despite its restriction to Dirac-$\delta$ perturbation (instantaneous in time), the formalism
correctly evaluates performance measures even in the physically relevant case of perturbations with finite, but not too long duration. 
One would naively guess that the most critical lines are those that are the most heavily loaded,
either relatively to their thermal limit or in absolute value. Quite surprisingly, we found that faults on lines
transmitting less than half of the heaviest line load in the network sometimes 
require more primary effort control or perturb the network's coherence more than lines with higher loads. 

Future works should investigate nodal $N-1$ faults, where a bus with all its connected lines is removed from the network. 
Another possible direction would be to consider inertialess nodes with first order dynamics, modeling
droop controlled inverters connecting PV productions to the grid.

\section*{Acknowledgment}
 We thank B.~Bamieh and F.~D\"orfler for useful discussions. This work was supported by the Swiss National Science Foundation
 under an AP Energy Grant.

\vspace{-1cm}

\begin{IEEEbiographynophoto}{Tommaso Coletta}
 Received the M.Sc. degree in physics and the Ph.D. degree in theoretical physics from the 
 Ecole Polytechnique F\'ed\'erale de Lausanne (EPFL), Lausanne, Switzerland in 2009 and 2013 respectively.
 He has been a Postdoctoral researcher at the Chair of Condensed Matter Theory at the Institute of Theoretical Physics of EPFL.
 Since 2014 he is a Postdoctoral researcher at the Engineering department of the University of 
 Applied Sciences of Western Switzerland, Sion, Switzerland working on complex networks and power systems.
 \end{IEEEbiographynophoto} 
\vspace{-1cm}
 \begin{IEEEbiographynophoto}{Philippe Jacquod}
 Philippe Jacquod received the Diplom degree in theoretical physics from the ETHZ, Z\"{u}rich, Switzerland, in 1992, and the PhD degree in 
 natural sciences from the University of Neuch\^{a}tel, 
 Switzerland, in 1997. He is a professor with the Energy and Environmental Engineering department,
 School of Engineering,
 University of Applied Sciences of Western Switzerland, Sion, Switzerland, with a joint appointment with the Department of Quantum Matter Physics, University of Geneva, Switzerland. From 2003 to 2005 he was an assistant professor with the 
 theoretical physics department, University of Geneva, 
 Switzerland and from 2005 to 2013 he was a professor 
 with the physics department, University of Arizona, Tucson, USA. His main research topics is in power systems and how they
 evolve as the energy transition unfolds. He is co-organizing an international conference series on related topics. He has published about 100 
 papers in international journals, books and conference proceedings. 
 \end{IEEEbiographynophoto}


\begin{thebibliography}{10}
\providecommand{\url}[1]{#1}
\csname url@samestyle\endcsname
\providecommand{\newblock}{\relax}
\providecommand{\bibinfo}[2]{#2}
\providecommand{\BIBentrySTDinterwordspacing}{\spaceskip=0pt\relax}
\providecommand{\BIBentryALTinterwordstretchfactor}{4}
\providecommand{\BIBentryALTinterwordspacing}{\spaceskip=\fontdimen2\font plus
\BIBentryALTinterwordstretchfactor\fontdimen3\font minus
  \fontdimen4\font\relax}
\providecommand{\BIBforeignlanguage}[2]{{%
\expandafter\ifx\csname l@#1\endcsname\relax
\typeout{** WARNING: IEEEtran.bst: No hyphenation pattern has been}%
\typeout{** loaded for the language `#1'. Using the pattern for}%
\typeout{** the default language instead.}%
\else
\language=\csname l@#1\endcsname
\fi
#2}}
\providecommand{\BIBdecl}{\relax}
\BIBdecl

\bibitem{Kur75}
Y.~Kuramoto, \emph{Lecture Notes in Physics}, vol.~39, pp. 420--422, 1975.

\bibitem{Pik01}
A.~Pikovsky, M.~Rosenblum, and J.~Kurths, \emph{Synchronization: A Universal
  Concept in Nonlinear Sciences}.\hskip 1em plus 0.5em minus 0.4em\relax
  Cambridge University Press, 2001.

\bibitem{Bialek08}
J.~Machowski, J.~W. Bialek, and J.~R. Bumby, \emph{Power system dynamics:
  stability and control.}\hskip 1em plus 0.5em minus 0.4em\relax John Wiley,
  2008.

\bibitem{AnnualEnOutlook}
``Power systems of the future: The case for energy storage, distributed
  generation, and microgrids,'' {\relax IEEE Smart Grid, Tech. Rep. Nov. 2012}.

\bibitem{Bac13}
S.~Backhaus and M.~Chertkov, ``Getting a grip on the electrical grid,''
  \emph{Phys. Today}, vol.~66, no.~5, pp. 42--48, 2013.

\bibitem{Rohden2014}
M.~Rohden, A.~Sorge, D.~Witthaut, and M.~Timme, ``Impact of network topology on
  synchrony of oscillatory power grids,'' \emph{Chaos: An Interdisciplinary
  Journal of Nonlinear Science}, vol.~24, no.~1, p. 013123, 2014.

\bibitem{Bamieh12}
B.~Bamieh, M.~R. Jovanovic, P.~Mitra, and S.~Patterson, ``Coherence in
  large-scale networks: Dimension-dependent limitations of local feedback,''
  \emph{IEEE Transactions on Automatic Control}, vol.~57, no.~9, pp.
  2235--2249, 2012.

\bibitem{Summers15}
T.~Summers, I.~Shames, J.~Lygeros, and F.~D{\"o}rfler, ``Topology design for
  optimal network coherence,'' in \emph{European Control Conference}.\hskip 1em
  plus 0.5em minus 0.4em\relax IEEE, 2015, pp. 575--580.

\bibitem{Siami14}
M.~Siami and N.~Motee, ``Systemic measures for performance and robustness of
  large-scale interconnected dynamical networks,'' in \emph{53rd Annual
  Conference on Decision and Control}.\hskip 1em plus 0.5em minus 0.4em\relax
  IEEE, 2014, pp. 5119--5124.

\bibitem{Fardad14}
M.~Fardad, F.~Lin, and M.~R. Jovanovic, ``Design of optimal sparse
  interconnection graphs for synchronization of oscillator networks,''
  \emph{IEEE Transactions on Automatic Control}, vol.~59, no.~9, pp.
  2457--2462, 2014.

\bibitem{Zhou96}
K.~Zhou, J.~Doyle, and K.~Glover, ``Robust and optimal control,'' 1996.

\bibitem{bamieh2013price}
B.~Bamieh and D.~F. Gayme, ``The price of synchrony: Resistive losses due to
  phase synchronization in power networks,'' in \emph{American Control
  Conference}.\hskip 1em plus 0.5em minus 0.4em\relax IEEE, 2013, pp.
  5815--5820.

\bibitem{Gayme15}
E.~Tegling, B.~Bamieh, and D.~F. Gayme, ``The price of synchrony: Evaluating
  the resistive losses in synchronizing power networks,'' \emph{IEEE
  Transactions on Control of Network Systems}, vol.~2, no.~3, pp. 254--266,
  2015.

\bibitem{Gayme16}
T.~W. Grunberg and D.~F. Gayme, ``Performance measures for linear oscillator
  networks over arbitrary graphs,'' \emph{IEEE Transactions on Control of
  Network Systems}, vol.~PP, no.~99, pp. 1--1, 2016.

\bibitem{Poola17}
B.~K. Poolla, S.~Bolognani, and F.~D\"{o}rfler, ``Optimal placement of virtual
  inertia in power grids,'' \emph{IEEE Transactions on Automatic Control},
  vol.~PP, no.~99, pp. 1--1, 2017.

\bibitem{Klein93}
D.~J. Klein and M.~Randi{\'c}, ``Resistance distance,'' \emph{Journal of
  Mathematical Chemistry}, vol.~12, no.~1, pp. 81--95, 1993.

\bibitem{Stephenson89}
K.~Stephenson and M.~Zelen, ``Rethinking centrality: Methods and examples,''
  \emph{Social Networks}, vol.~11, no.~1, pp. 1 -- 37, 1989.

\bibitem{Dorfler13Kron}
F.~D\"{o}rfler and F.~Bullo, ``Kron reduction of graphs with applications to
  electrical networks,'' \emph{IEEE Transactions on Circuits and Systems I},
  vol.~60, no.~1, pp. 150--163, 2013.

\bibitem{Itz80}
C.~Itzykson and J.~Zuber, \emph{Quantum Field Theory}.\hskip 1em plus 0.5em
  minus 0.4em\relax Mc Graw-Hill, 1980.

\bibitem{klein1997graph}
D.~J. Klein, ``Graph geometry, graph metrics and {W}iener,'' \emph{Commun.
  Math. Comput. Chem.}, no.~35, pp. 7--27, 1997.

\bibitem{xiao2003resistance}
W.~Xiao and I.~Gutman, ``Resistance distance and laplacian spectrum,''
  \emph{Theoretical Chemistry Accounts}, vol. 110, no.~4, pp. 284--289, 2003.

\bibitem{coletta17}
T.~Coletta and P.~Jacquod, ``Resistance distance criterion for optimal slack
  bus selection,'' \emph{arXiv preprint arXiv:1707.02845}, 2017.

\bibitem{paganini2017global}
F.~Paganini and E.~Mallada, ``Global performance metrics for synchronization of
  heterogeneously rated power systems: The role of machine models and
  inertia,'' in \emph{55th Annual Allerton Conference on Communication,
  Control, and Computing}, Oct 2017, pp. 324--331.

\bibitem{Wu16}
X.~Wu, F.~D\"orfler, and M.~R. Jovanovic, \emph{IEEE Transactions on Power
  Systems}, vol.~31, no.~3, pp. 2434--2444, 2016.

\bibitem{Manik2014}
D.~Manik, D.~Witthaut, B.~Sch{\"a}fer, M.~Matthiae, A.~Sorge, M.~Rohden,
  E.~Katifori, and M.~Timme, ``Supply networks: Instabilities without
  overload,'' \emph{The European Physical Journal Special Topics}, vol. 223,
  no.~12, pp. 2527--2547, Oct 2014.

\bibitem{coletta2016linear}
T.~Coletta and P.~Jacquod, ``Linear stability and the braess paradox in
  coupled-oscillator networks and electric power grids,'' \emph{Phys. Rev. E},
  vol.~93, p. 032222, 2016.

\bibitem{KouMachineReport}
\BIBentryALTinterwordspacing
G.~Kou, S.~W. Hadley, P.~Markham, and Y.~Liu, ``Developing generic dynamic
  models for the 2030 eastern interconnection grid,'' Oak Ridge National
  Laboratory, Tech. Rep., Dec. 2013. [Online]. Available:
  \url{http://www.osti.gov/scitech/}
\BIBentrySTDinterwordspacing

\bibitem{Pag18}
L.~Pagnier and P.~Jacquod, \emph{in preparation}, 2018.

\bibitem{Stewart}
G.~Stewart and J.~Sun, \emph{Matrix Perturbation Theory}.\hskip 1em plus 0.5em
  minus 0.4em\relax Academic Press, 1990.

\bibitem{Bozzo13}
E.~Bozzo and M.~Franceschet, ``Resistance distance, closeness, and
  betweenness,'' \emph{Social Networks}, vol.~35, no.~3, pp. 460 -- 469, 2013.

\bibitem{Berg81}
A.~Bergen and D.~Hill, ``A structure preserving model for power system
  stability analysis,'' \emph{IEEE Transactions on Power Apparatus and
  Systems}, vol. PAS--100, no.~1, pp. 25--35, 1981.

\bibitem{Sherman50}
J.~Sherman and W.~J. Morrison, ``Adjustment of an inverse matrix corresponding
  to a change in one element of a given matrix,'' \emph{Ann. Math. Statist.},
  vol.~21, no.~1, pp. 124--127, 1950.

\bibitem{Horn2012matrix}
R.~A. Horn and C.~R. Johnson, \emph{Matrix Analysis}, 2nd~ed.\hskip 1em plus
  0.5em minus 0.4em\relax New York, NY, USA: Cambridge University Press, 2012.

\bibitem{Kundur}
P.~Kundur, N.~J. Balu, and M.~G. Lauby, \emph{Power system stability and
  control}.\hskip 1em plus 0.5em minus 0.4em\relax McGraw-hill New York, 1994,
  vol.~7.

\bibitem{IEEEtestcase}
\BIBentryALTinterwordspacing
U.~of~Washington, ``Power systems test case archive,'' 1993. [Online].
  Available: \url{https://www2.ee.washington.edu/research/pstca}
\BIBentrySTDinterwordspacing

\bibitem{tyloo2017robustness}
M.~Tyloo, T.~Coletta, and P.~Jacquod, ``{Robustness of synchrony in complex
  networks and generalized Kirchhoff indices},'' \emph{Phys. Rev. Lett.}, vol.
  120, p. 084101, 2018.

\end{thebibliography}
\end{document}